\documentclass[num-refs]{wiley-article}
\papertype{Research Article}
\usepackage[utf8]{inputenc}
\title{Number Parsing at a Gigabyte per Second}
\author[1\authfn{1}]{Daniel Lemire}
\affil[1]{DOT-Lab Research Center, Universit\'e du Qu\'ebec (TELUQ), Montreal, Quebec, H2S 3L5, Canada}

\corraddress{Daniel Lemire, DOT-Lab Research Center, Universit\'e du Qu\'ebec (TELUQ), Montreal, Quebec, H2S 3L5, Canada}
\corremail{lemire@gmail.com}
\fundinginfo{Natural Sciences and Engineering Research Council of Canada, Grant Number: RGPIN-2017-03910}

\runningauthor{Daniel Lemire}
\usepackage[binary-units,per-mode=symbol]{siunitx}

\DeclareMathOperator{\round}{round}
\DeclareMathOperator{\ceiling}{ceiling}
\DeclareMathOperator{\remainder}{remainder}
\DeclareMathOperator{\floor}{floor}

\usepackage{booktabs}
\usepackage{listings}
\lstdefinestyle{customc}{%
  belowcaptionskip=1\baselineskip,
  breaklines=true,
  xleftmargin=\parindent,
  language=C,
  showstringspaces=false,
  basicstyle=\small\ttfamily,
  keywordstyle=\bfseries\color{green!40!black},
  numberstyle=\tiny,
  commentstyle=\itshape\color{purple!40!black},
  identifierstyle=\bfseries\color{black},
  stringstyle=\color{orange},
   morekeywords={uint64_t,uint32_t,__m256i,__m128i,simd8,uint8_t,UINT64_C},
}
\lstdefinestyle{custompython}{%
  belowcaptionskip=1\baselineskip,
  breaklines=true,
  xleftmargin=\parindent,
  language=Python,
  showstringspaces=false,
  basicstyle=\small\ttfamily,
  keywordstyle=\bfseries\color{green!40!black},
  numberstyle=\tiny,
  commentstyle=\itshape\color{purple!40!black},
  identifierstyle=\bfseries\color{black},
  stringstyle=\color{orange},
   morekeywords={uint64_t,uint32_t,__m256i,__m128i,simd8,uint8_t,UINT64_C},
}
\usepackage{subfloat,subfig}
\usepackage{graphicx}
\usepackage{algorithm}
\usepackage{algpseudocode}
\usepackage[table]{xcolor}  
\usepackage{microtype}
\newtheorem{techremark}{Technical Remark}

\newtheorem{numberedexample}{Example}

\algnewcommand{\IIf}[1]{\State\algorithmicif\ #1\ \algorithmicthen}
\algnewcommand{\EndIIf}{\unskip\ \algorithmicend\ \algorithmicif}

\usepackage{filecontents}
\usepackage{tikz}
\usepackage{pgfplots,pgfplotstable} 
\begin{document}

\maketitle
\begin{abstract}
With disks and networks providing gigabytes per second, parsing decimal numbers  from strings becomes a bottleneck. 
We consider the problem of parsing decimal numbers 
to the nearest binary floating-point value. 
The general problem requires variable-precision arithmetic.
However, we need at most 17~digits to represent 64-bit standard floating-point numbers (IEEE~754). Thus we can  represent the
decimal significand  with a single 64-bit word.  
By combining the significand and precomputed tables, 
we can compute the nearest floating-point number using as few as one or two 64-bit multiplications. 

 Our implementation can be several times faster than conventional functions present in standard C libraries on modern 64-bit systems (Intel, AMD, ARM and POWER9). 
Our work is available as open source software used by major systems such as Apache~Arrow and Yandex~ClickHouse. The Go standard library has adopted a version of our approach.


\end{abstract}

\section{Introduction}




Computers approximate real numbers as binary IEEE-754 floating-point numbers: an integer $m$ (the \emph{significand}\footnote{The use of the term \emph{mantissa} is discouraged  by IEEE~\cite{893287,cpe5418} and by Knuth~\cite{knuth2014art}.}) multiplied by 2 raised to an integer exponent $p$: $m \times 2^p$.
Most programming languages have a corresponding 64-bit data type and commodity processors provide
the corresponding instructions.  In  several mainstream programming languages (C, C++, Swift, Rust, Julia, C\#, Go), floating-point numbers adopt the 64-bit floating-point type by default. In JavaScript, all numbers are  represented using a 64-bit floating-point type, including integers---except maybe for the large integer type \texttt{BigInt}. 
There are other number types  beyond the standard binary IEEE-754 number types. For example,  Gustafson~\cite{gustafson2017end} has proposed Unums types, Microsoft promotes its Microsoft Floating Point (MSFP) types~\cite{darvish2020pushing}  and many programming languages support decimal-number data types~\cite{cowlishaw2001decimal}. However, they are not as ubiquitous.



Numbers are frequently serialized on disk or over a network as ASCII strings representing the value in decimal form (e.g., \texttt{3.1416}, \texttt{1.0e10}, \texttt{0.1}). 
It is generally impossible to find a binary IEEE-754   floating-point  number that matches exactly a decimal number. For example, the number \texttt{0.2} corresponding to $1/5$
can never be represented exactly as a binary floating-point number: its binary representation requires an infinite number of digits. Thus we must find the nearest available binary floating-point number. The nearest approximation to \texttt{0.2} using a standard 64-bit floating-point 
value is $\num{7205759403792794}\times    2 ^ {-55}$ or approximately $\num{0.20000000000000001110}$. The second nearest floating-point value is
 $\num{7205759403792793}\times    2 ^ {-55}$ or approximately $\num{0.19999999999999998335}$. 
  In rare cases, the decimal value would be exactly between two floating-point values. In such cases, the convention is that we \emph{round ties to even}: of the two nearest floating-point values, we choose the one with an even \emph{significand}. Thus, since \num{9000000000000000.5} falls at equal distance from  $\num{9000000000000000}\times    2 ^ {0}$ and 
$\num{9000000000000001}\times    2 ^ {0}$, we round it to $\num{9000000000000000}$. Meanwhile we round \num{9000000000000001.5} and \num{9000000000000002.5}  to \num{9000000000000002} and so forth.
 
 Finding the binary floating-point value that is closest to a decimal string can be computationally challenging.  
 Widely used number parsers fail to reach \SI{200}{\mebi\byte\per\second} on commodity processors (see Fig.~\ref{fig:serializedspeed}) whereas our disks and networks are capable of transmitting data at gigabytes per second (more than 5~times as fast).

If we write a 64-bit floating-point number as a string using a decimal significand with 17~significant digits, we can always parse it back exactly: starting from the ASCII string representing the number with a correctly rounded 17-digit decimal significand  and picking the nearest floating-point number, we retrieve our original number.  
Programmers and data analysts use no more than 17~digits in practice since there is no benefit to the superfluous digits if the number is originally represented as a standard binary  floating-point number. There are only some exceptional cases where more digits could be expected: e.g., 
\begin{enumerate}
\item when the value has been entered by a human being,
\item when the value was originally computed using higher-accuracy arithmetic,
\item when the original value was in a different number type,
\item when a system was poorly designed.
\end{enumerate}

We have that 64-bit unsigned integers can represent
all 19-digit non-negative integers since $10^{19} <2^{64}$. Given an ASCII string (e.g., \texttt{2.2250738585072019e-308}), we can 
 parse the decimal significand as a 64-bit integer and represent the number as $22250738585072019\times 10^{-308-16}$. It remains to convert it to a binary floating-point number. 
In this instance, we must divide $22250738585072019$ by $10^{308+16}$ and round the result: we show that we can solve such problems using an efficient algorithm (\S~\ref{sec:fastalgo}).
Further, when more than 19~digits are found, we may often be able to determine the nearest floating-point value from the most significant 19~digits (\S~\ref{sec:long}). 


\begin{table}[!tb]
    \centering
    \begin{tabular}{p{0.40\textwidth}p{0.45\textwidth}}
    \toprule
    term & notation\\
    \midrule     
    $m$ & binary significand; $m$ is a non-negative integer (often in   $[2^{52},2^{53})$)\\[0.2em]
    $p$ & binary exponent; $p$ is an integer, it can be negative\\[0.2em]
    non-negative floating-point number     & $m \times 2^p$ \\[0.2em]
        $w$ & decimal significand; $w$ is a non-negative integer (often in   $[0,2^{64})$) \\[0.2em]
    $q$ & decimal exponent; $q$ is an integer, it can be negative \\[0.2em]
     non-negative decimal number     & $w \times 10^q$\\[0.2em]
        rounding                     &     $\round(x)$ is an integer value nearest to $x $, with ties broken by rounding to the nearest even integer\\[0.2em]
        ceiling                     &    given $x\geq 0$,  $\ceiling(x)$ is the smallest  integer  no smaller than $x $\\[0.2em]
        floor                     &     given $x\geq 0$, $\floor(x)$ is the largest  integer  no larger than $x $\\[0.2em]
        integer quotient                     &    given non-negative integers $n,m$,  $n \div m = \floor (n /m)$ \\[0.2em] 
        integer remainder                     &     $\remainder(n , m)= n - m \times \floor (n /m)$ \\[0.2em] 
        trailing zeros & the positive integer $x$ has $k\in \mathbb{N}$ trailing zeros if and only if $x$ is divisible by $2^k$  \\
    \bottomrule
    \end{tabular}
    \caption{Notational conventions}
    \label{tab:notation}
\end{table}

%
Our main contribution is to show that we can reach high parsing speeds (e.g.,  \SI{1}{\gibi\byte\per\second}) on current 64-bit processors without sacrificing accuracy
by focusing on optimizing the common number-parsing scenario where we have no more than 19~digits. 
We make all of our software freely available.



\section{IEEE-754 Binary Floating-Point Numbers}

Many of the most popular programming languages and many of the most common processors support 64-bit and 32-bit  IEEE-754 binary floating-point numbers. See Table~\ref{tab:commmonieee}.
The IEEE-754 standard defines other  binary floating-point types (binary16, binary128 and binary256) but they are less common.

According to the IEEE-754 standard, a positive  \emph{normal} double-precision floating-point  number is a binary floating-point number where the 53-bit integer $m$ (the \emph{significand}) is in the interval $[2^{52},2^{53})$ while being interpreted as a number in $[1,2)$ by virtually dividing it by $2^{52}$, and 
where the 11-bit exponent $p$ ranges 
from \num{-1022} to
\num{+1023}~\cite{10.1145/103162.103163}. Such a double-precision number can represent all values between $2^{-1022}$ and up to but not including $2^{1024}$; these are the positive \emph{normal} values. Some values smaller than $2^{-1022}$ can be represented and are called \emph{subnormal} values:  they use a special exponent code which has the value $2^{-1022}$ and the significand  is then interpreted as a value in $[0,1)$.
A sign bit distinguishes negative and positive values.  A double-precision floating-point  number  uses 64~bits and is often called binary64.
The binary64 format can represent all decimal numbers made of a 15-digit significand from $\approx -1.8 \times 10^{308} $ to $\approx 1.8 \times 10^{308} $---except for the subnormal range ($\approx -5 \times 10^{-324} $ to $\approx 5 \times 10^{-324} $). Importantly, the reverse is not true: it is not sufficient to have 15~digits of precision to distinguish any two floating-point numbers: we may need up to 17~digits.


The single-precision floating-point  numbers are similar but span 32~bits (binary32). They are binary floating-point numbers where the 24-bit significand $m$ is in the interval $[2^{23},2^{24})$---considered as  value in $[1,2)$ after virtually dividing it by $2^{23}$---and 
where the 8-bit exponent $p$ ranges 
from \num{-126} to \num{+127}. We can represent all numbers between $2^{-126}$ up to, but not including, $2^{128}$; with special handling for some numbers smaller than $2^{-126}$ (subnormals).
The binary32 type can represent all decimal numbers made of a 6-digit significand.
If we serialize a 32-bit number using 9~digits, we can always parse it back exactly.

\begin{table}[!tb]
    \centering
    \begin{tabular}{cccc}
    \toprule
    name  & exponent bits & significand (stored) & decimal digits (exact)  \\\midrule
    
    binary64 &  11~bits & 53~bits (52~bits) & 15 (17)\\
    binary32 & 8~bits & 24~bits (23~bits) & 6 (9)\\
    \bottomrule
    \end{tabular}
    \caption{Common IEEE-754 binary floating-point numbers: 64~bits (binary64) and 32~bits (binary32). A single bit is reserved for the sign in all cases.}
    \label{tab:commmonieee}
\end{table}

\section{Related Work}
\label{sec:relatedwork}

Clinger~\cite{10.1145/93548.93557,10.1145/989393.989430} describes accurate decimal to binary conversion; he proposes a fast path using the fact that small powers of 10 can be represented exactly as floats. Indeed,  if we seek to convert the decimal number $1245\times 10^{14}$ to a binary floating-point number, we   observe  
that the number $10^{14}$ can be represented exactly as a 64-bit floating-point number because $10^{14} = 5^{14} \times 2^{14}$, and $5^{14} < 2^{53}$. Of course, the significand ($1245$) can also be exactly represented. Thus if the value $10^{14}$ is precomputed as an exact floating-point value, it remains to compute the product of $1245\times 10^{14}$.
The IEEE-754 specification requires that the result of an elementary arithmetic operation is correctly rounded.\footnote{Mainstream commodity processors (e.g., x64 and 64-bit ARM) have fast floating-point instructions with correct 64-bit and 32-bit rounding.} 
In this manner, we can immediately convert a decimal number to a 64-bit floating-point number  when it can be written as $w \times 10^{q}$ with  $-22\leq q\leq 22$ and  $w\leq 2^{53}$. We can extend Clinger's fast approach to 32-bit floating-point numbers with the conditions $-10\leq q\leq 10$ and $w\leq 2^{24}$.

Gay~\cite{gay1990correctly} improves upon Clinger's work in many respects. He provides a secondary fast path, for slightly larger powers. Indeed if we need to compute
 $w \times 10^{q}$ for some small integer $w$  for some integer $q>22$, we may   
decompose the problem as  $(w \times 10^{q-22}) \times 10^{22}$. If $q$ is  sufficiently small  so that
 $w \times 10^{q-22}$ is  less than $2^{53}$, then the computation is still  exact, and thus $(w \times 10^{q-22}) \times 10^{22}$ is also  exact.  Unfortunately, this approach is limited to decimal exponents $q\in (22,22+16)$ since
 $10^{16}>2^{53}$.
Gay contributes  a fast general decimal-to-binary implementation that is still in wide use: we benchmark against its most recent implementation in \S~\ref{sec:experiments}. The general strategy for decimal-to-binary conversion involves first finding quickly a close approximation, that is within a few floating-point numbers of the accurate value, and then to refine it using one or two more steps involving exact big-integer arithmetic.
Though there has been many practical attempts at optimizing number parsing~\cite{abseil}, we are not aware of improved follow-up work to Gay's approach in the scientific literature. 
%

A tangential problem is the conversion of binary floating-point numbers to decimal strings; the inverse of the problem that we are considering.  The binary-to-decimal problem has received much attention~\cite{10.1145/3192366.3192369,10.1145/3360595,10.1145/2837614.2837654,10.1145/249069.231397,10.1145/1806596.1806623,10.1145/989393.989431}.
Among other problems is the one of representing a  floating-point number using as few digits as possible so that the exact original value can be retrieved~\cite{10.1145/3360595}: we need between 1 and 17~digits.

\section{Parsing the String}
\label{sec:parsestring}

A  floating-point value may be encoded in different manners as a string. For example, \texttt{1e+1}, \texttt{10}, \texttt{10.0}, \texttt{10.}, \texttt{1.e1}, \texttt{+1e1}, \texttt{10E-01}  all represent the same value (10). There are different conventions and rules. For example, in JSON~\cite{rfc8259}, the following strings are invalid numbers: +1, 01, 1.e1. Furthermore, there are locale-specific conventions.

When parsing decimal numbers, our first step is to convert the string into a significand and an exponent.  Though details differ depending on the requirement, the general strategy we propose is as follows:

\begin{enumerate}

\item The string may  be explicitly delimited---we have the end point or a string length---or we may use a sentinel such as the null character.  The parser must not access characters outside the string range to avoid memory errors and security issues. 
\item It may be necessary to skip all leading white-space characters. In general, what constitutes a white-space character is locale-specific.
\item The number may begin with the `+' or the `-' character. Some formats may disallow the leading `+' character.
\item The significand is a sequence of digits (0,\ldots,9) containing optionally the decimal separator: the period character~`.' or a locale-specific equivalent. We must check that at least one digit was encountered: we thus forbid length-zero significands and significands made solely of the decimal separator. Some formats like JSON disallow a leading zero (only the zero value may begin with a zero) or an empty integer component (there must be digits before the decimal separator) or an empty fractional component (there must be digits after the decimal separator). To compute the significand, we may use a 64-bit unsigned integer $w$. We also record the beginning of the significand (e.g., as a pointer). We compute the digit value from the character using integer arithmetic: the digits have consecutive code point values in the most  popular character encodings (Unicode and ASCII define values from 48 to 57). We can similarly use the fact that the digits occupy consecutive code points to quickly check whether a character is a digit.
With each digit encountered, we can compute the running significand with a multiplication by ten followed by an addition ($w = 10 \times d  + v$ where $v$ is the digit value). The multiplication by ten is often optimized by compilers into efficient sequences of instructions. If there are too many digits, the significand may overflow 
which we can guard against by counting the number of processed digits: it is not necessary to guard each addition and multiplication against overflows. We can either be optimistic and later check whether an overflow was possible, or else we may check our position in the string, making sure that we never parse more than 19~digits, after omitting leading zeros. When the decimal separator is encountered, we record its position, but we otherwise continue computing the running significand. 
It is common to encounter many digits after the decimal separator. Instead of processing the digits one by one, we may check all at once whether a sequence of 8~digits is available and then update a single time the running significand---using a technique called SIMD within a register (SWAR)~\cite{fisher1998compiling}.  See Appendix~\ref{appendix:parsing-string}. If we find a sequence of eight digits, it can be beneficial to check again whether eight more digits (for a total of 16) can be found.

\item 
If there is a decimal separator, we must record the number of fractional digits, which we compute from the position of the decimal separator  and the end of the significand. E.g., if there are 12 digits after the decimal separator, then the exponent is $10^{-12}$. If there is no decimal separator, then the exponent is implicitly zero ($10^0$).
\item The significand may be followed by the letter `e' or the letter `E' in which case we need to parse the exponent if the scientific notation is allowed. Conversely, if the scientific notation is prescribed, we might fail if the exponent character is not detected. The parsing of the exponent proceeds much like the parsing of the significand except that no decimal separator is allowed.
An exceptional condition may occur if the exponent character is not followed by digits, accounting for the possible `+' and `-' characters. We may either fail, if the scientific notation is required, or we may decide to truncate the string right before the exponent character.
To avoid overflow with the exponent, we may  update it only if its absolute value is under some threshold: it makes no difference whether the exponent is \num{-1000} or \num{-10000}; whether it is \num{1000} or \num{10000}. 
The explicit exponent must be added to the exponent computed from the decimal separator. 
\item If the number of digits used to express the significand is less than 19, then we know that the significand cannot overflow. If the number of digits is more than 19, we may count the significant digits by omitting leading  zeros (e.g., \num{0.000123} has only three significant digits). Finally, if an overflow cannot be dismissed, we may need to parse using a higher-precision code path (\S~\ref{sec:long}).
\end{enumerate}

There are instances when we can quickly terminate the computation after decoding the 
decimal significand and its exponent. If the significand is zero or the exponent is very small then the number must be zero. If the significand is non-zero but the exponent is very large then we have an infinite value ($\pm \infty$).


\section{Fast Algorithm}
\label{sec:fastalgo}

A fast algorithm to parse floating-point numbers might start  by processing
the ASCII string (see \S~\ref{sec:parsestring}) to find a decimal significand and a
decimal exponent. If the number of digits in the  significand  is less than 19, then
our approach is applicable. 
(see \S~\ref{sec:long}).
However, before we apply our algorithm, 
 we use Clinger's  fast path~\cite{10.1145/93548.93557,10.1145/989393.989430}, see \S~\ref{sec:relatedwork}. Even though it adds an additional branch at the beginning, it is an inexpensive code path when it is applicable, implying a single floating-point multiplication or division. We can implement it efficiently. We check whether the decimal power is within the allowable interval ($q\in [-22,22]$ in the 64-bit case, $q\in [-10,10]$ in the 32-bit case) and whether the absolute value of the decimal significand is in the allowable interval ($[0,2^{53}]$ in the 64-bit case or $[0,2^{24}]$ in the 32-bit case). When these conditions are encountered, we losslessly convert the decimal significand to a floating-point value, we lookup the precomputed power of ten $10^{|q|}$ and we multiply (when $q 
\geq 0$) or divide ($q < 0$) the converted significand.\footnote{The case with negative exponents where a division is needed requires some  care on systems where the division of two floating-point numbers is not guaranteed to round the nearest floating-point value: when such a system is detected, we may  limit the fast path to positive decimal powers.}  Gay~\cite{gay1990correctly} proposes an extended fast path that covers a broader range of decimal exponents, but with more stringent conditions on the significand. We do not make use of this secondary fast path. It adds additional branching and complexity for relatively little gain in our context.

In particular, Clinger's fast path covers all integer values in $[0, 2^{53}]$ (64-bit case). We could also add an additional fast path specifically for integers.  We can readily identify such cases because the decimal exponent is zero: $w\times 10^0$. We can rely on the fact that the IEEE standard specifies that conversion between integer and floating-point be correctly rounded~\cite{10.1145/103162.103163}. Thus a cast from an integer value to a floating-point value is often all that is needed. It may often require nearly just a single instruction (e.g., \texttt{cvtsi2sd} under x64 and \texttt{ucvtf} under ARM). However, we choose to disregard this potential optimization because the gains are  modest while it increases the complexity of the code.

We must then handle the general case, after the application of Clinger's fast path.
We formalize our approach with Algorithm~\ref{algo:fancytotalalgo}. This concise algorithm can handle rounding, including ties to even, subnormal numbers  and infinite values. We specialize the code for positive numbers, but negative numbers are handled by flipping the sign bit in the result. As the pseudo-code suggests, it can be implemented in a few lines of code. The algorithm always succeeds unless we have  large or  small decimal
exponents ($q \notin [-27,55]$) in which case we may need to fall back on a higher-precision approach in uncommon instances. The algorithm relies on a precomputed table of 128-bit values $T[q]$ for decimal exponents $q\in [-342,308]$ (see Appendix~\ref{appendix:table}).
\begin{itemize}
    \item In lines~\ref{line:guard1}~and~\ref{line:guard2}, we check for very large or very small decimal exponents as well as for zero decimal significands. In such cases, the result is always either zero or infinity.
    \item In lines~\ref{line:norm1}~and~\ref{line:norm2}, we normalize the decimal significand $w$ by shifting it so that  $w\in [2^{63}, 2^{64})$.
    \item We must convert the decimal significand $w$ into the binary significand $m$.
    We have that $w \times 10^q = w \times 5^q \times 2^q   \approx m \times 2^p$ so we must estimate $ w \times 5^q$. At line~\ref{line:multiplication}, we multiply the normalized significand $w$ by the 128-bit value $T[q]$ using one or two 64-bit multiplications. Intuitively, the product $w \times T[q]$  approximates $w \times 5^q$ after shifting the result. We describe this step in \S~\ref{sec:mostsig} and \S~\ref{sec:pospower} for positive decimal exponents ($q\geq 0$), and in \S~\ref{sec:divisionpowerfive} for negative decimal exponents. We have a 128-bit result $z$.
    \item At line~\ref{line:failure}, we check for failure,  requiring the software to fall back on a higher-precision approach. It corresponds to the case where we failed to provably approximate $ w \times 5^q$ to a sufficient degree. In practice, it is unlikely and only ever possible if $q\notin [-27,55]$.
    \item At line~\ref{line:binarysignificand}, we compute the expected binary significand with one extra bit of precision (for rounding) from the product $z$.
    \item At lines~\ref{line:binaryexponent1}~and~\ref{line:binaryexponent2}, we compute the expected binary exponent. We justify this step in \S~\ref{sec:expo}.
    \item At line~\ref{line:binaryexponenttoosmall}, we check whether the binary exponent is too small. When it is too small, the result is zero.
    \item At line~\ref{line:subnormal}, we check whether we have a subnormal value when the binary exponent is too small. See \S~\ref{sec:subnormals}.
    \item At line~\ref{line:smallerthanone}, we handle the case where we might have a value that is exactly between two binary floating-point numbers. We describe this step generally in \S~\ref{sec:roundtoeven} where we show that subnormal values cannot require rounding ties. We describe it specifically in \S~\ref{sec:roundevenpositive} for the positive-exponent case ($q\geq 0$) and in \S~\ref{sec:nearone} in the negative-exponent case. 
    Intuitively, we identify ties when the product $z$ from which we extracted the binary significand $m$ contains many trailing zeroes after ignoring the least significant bit. 
    We need to be concerned  when we would (at line~\ref{line:round}) round up from an even value: we adjust the value to prevent rounding up. 
    \item At line~\ref{line:round}, we round the binary significand. At line~\ref{line:roundoverflow}, we handle the case where rounding up caused an overflow, in which case we need to increment the binary exponent. At line~\ref{line:inf}, we handle the case where the binary exponent is too large and we have an infinite value.
\end{itemize}
We show that the algorithm is correct by examining each step in the following sections. We assess our algorithm experimentally in \S~\ref{sec:experiments}.

\begin{algorithm}[!tbh]
\begin{algorithmic}[1]
\Require  an integer $w \in [0,10^{19}]$ and an integer exponent $q$
\Require  a table $T$ containing 128-bit reciprocals and truncated powers of five  for all powers from $-342$ to $308$ (see Appendix~\ref{appendix:table})
\IIf{$w = 0$ or $q < -342$\label{line:guard1}} \textbf{Return} 0 \EndIIf
\IIf{$q > 308$\label{line:guard2}} \textbf{Return} $\infty$ \EndIIf
\State \label{line:norm1}$l \leftarrow$ the number of leading zeros of $w$ as a 64-bit (unsigned) word
\State \label{line:norm2}$w \leftarrow 2^l \times w$ \Comment{Normalize the decimal significand}
\State  \label{line:multiplication} Compute the 128-bit truncated product  $z \leftarrow (T[q] \times w) \div 2^{64}$, stopping after one 64-bit multiplication if the most significant 55~bits (64-bit) or 26~bits (32-bit)  are provably exact.
\IIf{$z \bmod 2^{64} = 2^{64}-1$ and $q \notin [-27,55]$ \label{line:failure} } \textbf{Abort} \EndIIf{}
\State \label{line:binarysignificand}$m \leftarrow$ the most significant 54~bits (64-bit) or 25~bits (32-bit) of the product $z$, not counting the eventual leading zero bit
\State \label{line:binaryexponent1}$u\leftarrow z \div 2^{127}$ value of the most significant bit of $z$
\State \label{line:binaryexponent2}$p \leftarrow ((217706 \times q) \div 2^{16}) + 63 - l + u$ \Comment{Expected binary exponent}
\IIf{$p \leq -1022-64$ (64-bit) or $p \leq -126-64$ (32-bit)\label{line:binaryexponenttoosmall}}
\textbf{Return} 0
\EndIIf
\If{\label{line:subnormal}$p \leq -1022$ (64-bit) or $p \leq -126$ (32-bit)}\Comment{Subnormals}
\State $s\leftarrow -1022 - p +1$ (64-bit) or $s\leftarrow -126 - p +1$ (32-bit)
\State $m \leftarrow m \div 2^s$ and 
 $m \leftarrow m + 1$ if $m$ is odd (round up), and $m \leftarrow m \div 2$
\State \textbf{Return} $m \times 2^{p} \times 2^{-52}$ (64-bit) or  $m \times 2^{p} \times 2^{-23}$ (32-bit case)
\EndIf
\If{\label{line:smallerthanone}$z \bmod 2^{64} \leq 1$ and $m$ is odd and $m\div 2$ is even and ($q \in [-4,23]$ (64-bit) or $q \in [-17,10]$ (32-bit))} \Comment{Round ties to even}
\State \textbf{if} $(z \div 2^{64})/m$ is a power of two \textbf{then}  $m\leftarrow m-1$ \Comment{Will not round up}
\EndIf
\State \label{line:round} $m \leftarrow m + 1$ if $m$ is odd; followed by 
 $m \leftarrow m \div 2$ \Comment{Round the binary signficand}
\IIf{\label{line:roundoverflow}$m = 2^{53}$ (64-bit) or $m = 2^{24}$ (32-bit)}
$m = m \div 2$; $p \leftarrow p + 1$
\EndIIf
\IIf{\label{line:inf}$p>1023$ (64-bit) or $p>127$ (32-bit)}
\textbf{Return} $\infty$
\EndIIf
\State \textbf{Return} $m \times 2^{p} \times 2^{-52}$ (64-bit) or  $m \times 2^{p} \times 2^{-23}$ (32-bit case)
\end{algorithmic}
\caption{%
Algorithm to compute the binary floating-point number nearest to a decimal floating-point number $w \times 10^q$. We give just one algorithm for both the 32-bit and 64-bit cases. For negative integers, we need to  negate the result.\label{algo:fancytotalalgo}}
\end{algorithm}

\section{Exact Numbers and Ties}
\label{sec:roundtoeven}

We seek to approximate a decimal floating-point number  of the form $w \times 10^q$ using a binary decimal floating-point number of the form $m \times 2^p$. Sometimes, there is no need to approximate since an exact representation is possible.
That is, we have that  $w \times 10^q = m \times 2^p$ or, equivalently, $w \times 5^q \times 2^q = m \times 2^p$. In our context, we refer to these numbers as \emph{exact numbers}. We seek to better identify when they can occur.

\begin{itemize}
\item When $q\geq 0$,  we  have  $m = w \times 5^q \times 2^q \times 2^ {-p}$ so that $m$ is divisible by $5^q$.
In the 64-bit case, we have that $m<2^{53}$; and in the 32-bit case, we have that $m<2^{24}$. Thus we have, respectively, $5^q <2^{53}$ and $5^q <2^{24}$. These inequalities become $q\leq 22$ and $q\leq 10$. For example, we have that $1 \times 10^{22}$ is an exact 64-bit number while $1 \times 10^{23}$ is not. 
\item When $q<0$, we have that $w = 5^{-q} \times 2^{-q}  \times 2^{p}  \times m$. We have that $5^{-q}$ divides $w$. If we assume that $w<2^{64}$ then we have that
$5^{-q} < 2^{64}$
or $q\geq -27 $. For example, the number $7450580596923828125 \times 10 ^{-27}$ is the smallest  exact 64-bit number.   It follows that no exact number is sufficiently small to qualify as a subnormal value: the largest subnormal number  has a small decimal power (e.g., $\approx 10^{-38}$ in the 32-bit case). 

\end{itemize}
Thus we have that exact numbers must be of the form $w \times 10^q$ with $q\in [-27, 22]$ (64-bit case) or $q\in [-27, 10]$ (32-bit case) subject to the constraint that the decimal significand can be stored in a 64-bit value. Yet floating-point numbers occupy a much wider range (e.g., from $4.9\times 10 ^{-324}$ to $1.8\times 10 ^{308} $). In other words, exact numbers are only possible when the decimal exponent is near zero.

%


To find  the nearest floating-point number when parsing, it is almost always sufficient to round
to the nearest value without  an exact computation. However, when the number
we are parsing might fall exactly between two numbers, more care is needed. The IEEE-754 standard recommends that
we round to even. We may need an exact
computation to apply the round-ties-to-even strategy.\footnote{We focus solely on 
rounding ties to even, as it is  ubiquitous. However, our approach could be extended to other rounding modes.}
The sign can be ignored when rounding ties to even: if a value is exactly between the two nearest floating-point numbers and they have different signs, then the  midpoint value must be zero, by symmetry.

It may seem that we could generate many cases where we fall exactly between
two floating-point numbers. Indeed, it suffices to take any floating-point number that is not the largest one, and then take the next largest floating-point number. From these two numbers, we pick a number that is right in-between and we have a number that requires rounding to even. However, for such a half-way number to be a concern to us, it must be represented exactly in decimal form using a small number of decimal digits.
In particular,  the decimal significand must be divisible by $5^{-q}$ and yet must be no larger than $2^{64}$.
It implies that the decimal exponent cannot be too small ($q\geq -27$). It also  implies that the nearby binary floating-point values are normal numbers. 

Let us formalize the analysis.
A mid-point between two floating-point numbers,
$m \times 2^{p}$ and $(m+1) \times 2^{p}$,
can be written as $(2m+1) \times 2^{p-1}$. 
Assume that both numbers $m \times 2^{p}$ and $(m+1) \times 2^{p}$ can be represented exactly using a standard floating-point type:
for 64-bit floating point numbers, it implies that $m + 1<2^{53} $; for 32-bit numbers   it implies $m + 1 <2^{24}$.
We only need rounding if  $(2m+1) \times 2^{p-1}$ cannot be represented.
There are two reasons that might explain why a number cannot be represented. Either it requires a power of two that is too small or too large, or else its significand requires too many bits.
Because the value is exactly between two  numbers that can be represented, we know that it is not outside the bounds of the power of two. Thus the significand must require too many bits. Furthermore, both $m \times 2^{p}$ and $(m+1) \times 2^{p}$ must be normal numbers.
Hence the fact that $2m + 1$ has too many bits implies that $2m + 1 \in (2^{53}, 2^{54}] $ for 64-bit floating point numbers and that $2m + 1 \in (2^{24}, 2^{25}]$ for 32-bit numbers.



\begin{itemize}
\item When $q\geq 0$,  we have that $5^q \leq 2m+1$.
In the 64-bit case, we have  $5^q \leq 2m+1 \leq 2^{54}$ or $q \leq 23$.
In the 32-bit case, we have  $5^q \leq 2m+1 \leq 2^{25}$ or $q \leq 10$.
\item 
When $q<0$, we have $w  \geq  (2m+1) \times 5^{-q}$. We must have that $w<2^{64}$ so
$(2m+1) \times 5^{-q} < 2^{64}$.
 We have that $2m+1>2^{53}$ (64-bit case) or $2m+1>2^{24}$ (32-bit case). Hence, we must have $2^{53} \times 5^{-q} < 2^{64}$ (64-bit) and
 $2^{24} \times 5^{-q} < 2^{64}$
  (32-bit).
Hence we have  $5^{-q} < 2^{11}$ or $q\geq -4$  (64-bit case) and
$5^{-q} < 2^{40}$  or $q\geq -17$  (32-bit case).
\end{itemize}
Thus we have that we only need to round ties to even when  we have that 
$q\in [-4,23]$ (in the 64-bit case) or $q\in [-17,10]$ (in the 32-bit case). In both cases, the power of five ($5^{|q|}$) fits in a 64-bit word.

\section{Most Significant Bits of a Product}
\label{sec:mostsig}
When converting decimal values to binary values, we may need to multiply or divide by large powers of ten (e.g., $10^{300}= 5^{300} 2^{300}$). Mainstream processors  compute the 128-bit product of two 64-bit integers using one or two fast instructions: e.g., with the single instruction \texttt{imul} (x64 processors) or  two instructions  \texttt{umulh} and \texttt{mul} (aarch64 processors).
However, we cannot represent an integer like $5^{300}$ using a single 64-bit integer. We may represent such large integers using multiple 64-bit words, henceforth a \emph{multiword integer}. 

We may compute the product between two multiword integers starting from the least significant bits. Thus if we are multiplying an integer that requires a single machine word $w$ with an integer that requires $n$~machine words, we can use $n$~64-bit multiplications starting with a multiplication between the word $w$ and the least significant word of the other integer, going up to the most significant words. See Algorithm~\ref{algo:standardalgo}.

\begin{algorithm}
\begin{algorithmic}[1]
\Require  an integer $w \in (0,2^{64})$ 
\Require a positive integer $b$ represented as $n$~words ($n>0$) $b_0, b_1, \ldots, b_{n-1}$ such that $b=\sum_{i=0}^{n-1} b_i 2^{64 i}$.
\State Allocate $n+1$ words $u_0, u_1, \ldots, u_n$
\State $p \leftarrow w \times b_0$
\State $u_0 \leftarrow p \bmod 2^{64}$
\State $r \leftarrow p \div 2^{64}$
\For{$i = 1, \ldots, n-1$}
 \State $p \leftarrow w \times b_i$ \Comment{$p \leq (2^{64}-1)^2$}
 \State $p \leftarrow p + r$ \Comment{$p \leq 2^{128} - 2^{64} + 1$}
 \State $u_i \leftarrow p \bmod 2^{64}$
\State $r \leftarrow p \div 2^{64}$
\EndFor
 \State $u_n \leftarrow r$
 \State \textbf{Return}: The result of the multiplication as an $n+1$-word $u$ such that $u=\sum_{i=0}^{n} u_i 2^{64 i}$.
\end{algorithmic}
\caption{%
Conventional algorithm to compute the product of a single-word integer and a multiple-word integer. \label{algo:standardalgo}}
\end{algorithm}

Such a conventional algorithm is inefficient when we only need to approximate the product. For example, maybe we only want the most significant word of the product and we would like to do the computation using only one or two multiplications.
Thankfully it is often possible in practice. Such partial multiplications are sometimes called \emph{truncated} multiplications~\cite{hars2006applications} and the result is sometimes called a \emph{short} product~\cite{fousse2007mpfr,krandick1993efficient,mulders2000short}.

\begin{algorithm}[!tbh]
\begin{algorithmic}[1]
\Require  an integer $w \in (0,2^{64})$ 
\Require a positive integer $b$ represented as $n$~words ($n>0$) $b_0, b_1, \ldots, b_{n-1}$ such that $b=\sum_{i=0}^{n-1} b_i 2^{64 i}$.
\State a desired number of exact words $m\in (0,n+1]$
\State Allocate $n+1$ words $u_0, u_1, \ldots, u_n$
\State $p \leftarrow w \times b_{n-1}$
\State $u_{n-1} \leftarrow p \bmod 2^{64}$
\State $u_{n} \leftarrow p \div 2^{64}$
\If{$m = 1$ and $u_{n-1}<2^{64}-w$}\label{line:one}
\State \textbf{Return}:$u_{n}$ \Comment{Stopping condition}
\EndIf
\For{$i = n-2, n-1,  \ldots, 0$}
 \State $p \leftarrow w \times b_i$
 \State $u_{i} \leftarrow p \bmod 2^{64}$
 \If{$u_{i+1} + ( p \div 2^{64})\geq 2^{64}$}
\State add 1 to $u_{i+2}$, if it exceeds $2^{64}$, set it to zero and add 1 to $u_{i+3}$ and so forth up to $u_n$ potentially
\EndIf
\State $u_{i+1} \leftarrow (u_{i+1} + ( p \div 2^{64}))\bmod 2^{64}$
 \If{$m\leq n - i$ and $u_{i}<2^{64}-w$}
 \State \textbf{Return}:$u_{u-m+1},\ldots , u_{n}$ \Comment{Stopping condition}
\EndIf
 \If{$m< n - i$ and $u_{i}<2^{64}-1$}
 \State \textbf{Return}:$u_{n-m+1},\ldots , u_{n}$ \Comment{Stopping condition}
\EndIf
\EndFor
 \State \textbf{Return}:$u_{n-m+1},\ldots , u_{n}$.
\end{algorithmic}
\caption{%
Algorithm to compute the $m$~most significant words of the product  of a single-word integer and a multiple-word integer.  \label{algo:fancyalgo}}
\end{algorithm}

Suppose that we have computed the product of the single-word integer ($w$) with the $k$~most significant words of the multiword integer: we have computed the $k+1$~words of the product $w \times (\sum_{i=n-k}^{n-1} b_i 2^{64 i}) \div 2^{64(n-k)}$. Compared
with the $k+1$~most significant words of the full product
$(w \times (\sum_{i=0}^{n-1} b_i 2^{64 i})) \div 2^{64(n-k)}$, we are possibly underestimating because we omit the contribution of the product between
the word $w$ and the least $n-k$~significant words. 
These less $n-k$~significant words have maximal value $2^{64(n-k)} -1$. Their product with the word $w$ is thus at most
$2^{64(n-k)}w -w$ and their contribution to the most significant 
words is at most $(2^{64(n-k)}w -w) \div 2^{64(n-k)} = w-1$.
Hence if the least significant computed word that  is no larger than $2^{64}- w + 1$, then all computed words are exact except maybe for that least significant one.  Our short product matches the full product. Thus we have \emph{stopping condition}: if we only want the  $k$~most significant words of the product, we can compute the $k+1$~most significant words from the $k$~most significant words of the multiword integers, and stop if the least significant word of the product is no larger than $2^{64}- w + 1$.  This stopping condition is most useful if $w$ is small ($w \ll 2^{64}$).
We also have another stopping condition that is more generally useful. Even if the least significant word is larger than $2^{64}- w + 1$, then the second least significant word needs to be incremented by one in the worst case. If the second least significant word is not $2^{64}-1$, then all other more significant words are exact. That is,  we have $k-2$~exact most significant words if the second last of our $k+1$~most significant words is not $2^{64}-1$.

By combining these two conditions, we  rarely have to compute more than 3~words using two multiplications to get the exact value of the most significant word.
Algorithm~\ref{algo:fancyalgo} presents a general algorithm.  We can stop maybe even earlier if we need even less than the most significant word, say $t$~bits. Indeed, unless all of the less significant bits in the computed most significant word have value 1, then an overflow (+1) does not affect the most significant bits of the most significant word. Thus the condition $u_{n-1}<2^{64}-w$  (line~\ref{line:one} in Algorithm~\ref{algo:fancyalgo}) can be replaced by ($u_{n-1}<2^{64}-w$ or $u_n \bmod 2^{64-t} \not  = 2^{64-t}-1$) if we only need $t$~exact bits of the product.

\section{Multiplication by Positive Powers of Five}
\label{sec:pospower}
When parsing a decimal number, we follow the general strategy of first identifying the non-negative decimal significand $w$ and its corresponding exponent $q$. We then seek to convert $w\times 10^{q}$ to a  binary floating-point number. For example, given the mass of the Earth in kilograms as  $5.972\times 10^{24}$, we might parse it first as
$5972 \times 10^{21}$. Our goal is to represent it as a nearest binary  floating-point number such as
$5561858415603638  \times  2 ^ {30}$. 
The sign bit is handled separately.


The largest integer we can represent with a 64-bit floating-point number is  $\approx 1.8 \times 10^{308}$. Thus, when processing numbers of the form $w \times 10^q$ for non-negative powers of $q$, we only have to worry about $q \in [0, 308]$. If any larger value of $q$ is found and the decimal significand is non-zero ($w>0$), the result is an infinity value (either $+\infty$ or $-\infty$).

\begin{figure}\centering
\includegraphics[width=0.49\textwidth]{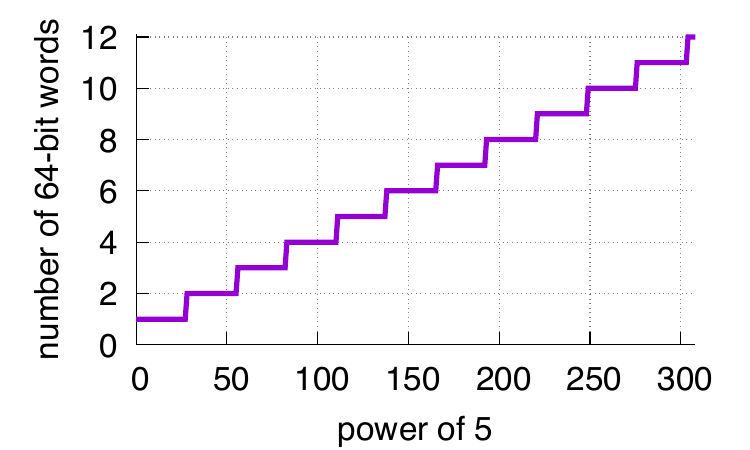}
\caption{\label{fig:words}Number of 64-bit words necessary to represent $5^q$ exactly for  $q\in [0,308]$. }
\end{figure}

Our first step is to expand the exponent: $w \times 10^q = w  \times 2^q \times 5^q $. Thus we seek to compute $w \times 5^q$. Though the number $5^{308}$ may appear large, it only requires twelve 64-bit~words since 
$5^{308}< 2^ {64\times 12}$. 
Storing all of these words require about \SI{15}{\kibi\byte}: we need between 1~and~12~words  per power (see Fig.~\ref{fig:words}).
Thus it could be practical to memoize all of the exponents $5^q$ for $q=1, \ldots, 308$.
In \S~\ref{sec:mostsig}, we compute the most significant word of a product with few words of the multiword integer: it follows that even if we need, in the worst case \SI{15}{\kibi\byte} of storage, an actual implementation may touch only a fraction of that memory. Thus it may be more efficient to only store  two 64-bit words  per power, as long as we can fall back  on a higher-precision approach. We then use
only \SI{5}{\kibi\byte} of storage (three times less).
Effectively, given a large power of five, we  store 128~bits of precision. 
We truncate the result, discarding its less significant words.




It could be that the most significant word of the product  $w \times 5^q$ contains many leading zeros. 
Yet we desire a normalized result, without leading zeros for efficiency and simplicity reasons.

\begin{itemize}

\item We  store the powers of five in a format where the most significant bit of the most significant word is 1: hence, we shift the power of five adequately: e.g., instead of storing $5^2$, we store $5^{2} \times 2^{59}$. See Table~\ref{tab:pospower}. 

\item Further,  we shift $w$ by an adequate power of two so that the 63$^\textrm{rd}$~bit has value 1.  It is always possible as $w$ is non-zero; the case when $w = 0$ is handled separately. Normalizing numbers such that they have no leading zeros  is inexpensive: the number of leading zeros can be computed using a single instruction on popular processors (\texttt{clz} on ARM,  \texttt{lzcnt} on x64). 
\end{itemize}

\begin{table}[tbh!]
    \centering
     \rowcolors{2}{gray!25}{white}
    \begin{tabular}{ccccc}
    \toprule
    $q$ & first word     & $q$ & first word & second word \\  
    \midrule     
0 & \texttt{8000000000000000}  & 28 & \texttt{813f3978f8940984} & \texttt{4000000000000000} \\
1 & \texttt{a000000000000000}  & 29 & \texttt{a18f07d736b90be5} & \texttt{5000000000000000} \\
2 & \texttt{c800000000000000}  & 30 & \texttt{c9f2c9cd04674ede} & \texttt{a400000000000000} \\
3 & \texttt{fa00000000000000}  & 31 & \texttt{fc6f7c4045812296} & \texttt{4d00000000000000} \\
4 & \texttt{9c40000000000000}  & 32 & \texttt{9dc5ada82b70b59d} & \texttt{f020000000000000} \\
5 & \texttt{c350000000000000}  & 33 & \texttt{c5371912364ce305} & \texttt{6c28000000000000} \\
6 & \texttt{f424000000000000}  & 34 & \texttt{f684df56c3e01bc6} & \texttt{c732000000000000} \\
7 & \texttt{9896800000000000}  & 35 & \texttt{9a130b963a6c115c} & \texttt{3c7f400000000000} \\
8 & \texttt{bebc200000000000}  & 36 & \texttt{c097ce7bc90715b3} & \texttt{4b9f100000000000} \\
9 & \texttt{ee6b280000000000}  & 37 & \texttt{f0bdc21abb48db20} & \texttt{1e86d40000000000} \\
10 & \texttt{9502f90000000000}  & 38 & \texttt{96769950b50d88f4} & \texttt{1314448000000000} \\
11 & \texttt{ba43b74000000000}  & 39 & \texttt{bc143fa4e250eb31} & \texttt{17d955a000000000} \\
12 & \texttt{e8d4a51000000000}  & 40 & \texttt{eb194f8e1ae525fd} & \texttt{5dcfab0800000000} \\
13 & \texttt{9184e72a00000000}  & 41 & \texttt{92efd1b8d0cf37be} & \texttt{5aa1cae500000000} \\
14 & \texttt{b5e620f480000000}  & 42 & \texttt{b7abc627050305ad} & \texttt{f14a3d9e40000000} \\
15 & \texttt{e35fa931a0000000}  & 43 & \texttt{e596b7b0c643c719} & \texttt{6d9ccd05d0000000} \\
16 & \texttt{8e1bc9bf04000000}  & 44 & \texttt{8f7e32ce7bea5c6f} & \texttt{e4820023a2000000} \\
17 & \texttt{b1a2bc2ec5000000}  & 45 & \texttt{b35dbf821ae4f38b} & \texttt{dda2802c8a800000} \\
18 & \texttt{de0b6b3a76400000}  & 46 & \texttt{e0352f62a19e306e} & \texttt{d50b2037ad200000} \\
19 & \texttt{8ac7230489e80000}  & 47 & \texttt{8c213d9da502de45} & \texttt{4526f422cc340000} \\
20 & \texttt{ad78ebc5ac620000}  & 48 & \texttt{af298d050e4395d6} & \texttt{9670b12b7f410000} \\
21 & \texttt{d8d726b7177a8000}  & 49 & \texttt{daf3f04651d47b4c} & \texttt{3c0cdd765f114000} \\
22 & \texttt{878678326eac9000}  & 50 & \texttt{88d8762bf324cd0f} & \texttt{a5880a69fb6ac800} \\
23 & \texttt{a968163f0a57b400}  & 51 & \texttt{ab0e93b6efee0053} & \texttt{8eea0d047a457a00} \\
24 & \texttt{d3c21bcecceda100}  & 52 & \texttt{d5d238a4abe98068} & \texttt{72a4904598d6d880} \\
25 & \texttt{84595161401484a0}  & 53 & \texttt{85a36366eb71f041} & \texttt{47a6da2b7f864750} \\
26 & \texttt{a56fa5b99019a5c8}  & 54 & \texttt{a70c3c40a64e6c51} & \texttt{999090b65f67d924} \\
27 & \texttt{cecb8f27f4200f3a}  & 55 & \texttt{d0cf4b50cfe20765} & \texttt{fff4b4e3f741cf6d} \\
    \bottomrule
    \end{tabular}
    \caption{Most significant bits in hexadecimal form of the powers $5^q$ for some exponents $q$. The values are normalized by multiplication with a power of two so that the most significant bit is always 1. For each power, the first word represent the most significant 64~bits, the second word the next most significant 64~bits. For $q\leq 27$, the second word is made of zeros and omitted. In practice, we need to materialize this table for up to $q=308$ to cover all of the relevant powers of five for 64-bit numbers.}
    \label{tab:pospower}
\end{table}
Hence, because of our normalization, the computed product is at least as large as $2^{63}\times 2^{63}=2^{126}$. That is, as a 128-bit value, it  has at most one leading 0. We can check for this case and shift the result by one bit accordingly.
To compute the product of  the normalized decimal significand with normalized words representing the power of five (see in Table~\ref{tab:pospower}), we use up to two multiplications.
\begin{itemize}
    \item
When $5^q<2^{64}$, then a single multiplication always provide an exact result. In particular, it implies that whenever we need to round ties to even, we have the exact product (see \S~\ref{sec:roundevenpositive}).
\item   We can always do just one multiplication as long as  it provides
the number of significant bits of the floating-point standard (53~bits for 64-bit numbers and 24~bits for 32-bit numbers) plus one extra bit to determine the rounding direction, and one more bit to handle the case where the computed product has a leading zero. It suffices to check the least significant bits of the most significant 64~bits and verify that there is at least one non-zero bit out of 9 bits for 64-bit numbers and out of 38~bits for 32-bit numbers. 
\item  When that is not the case, we execute a second multiplication. This second multiplication is always sufficient to compute the product exactly as long as $q\leq 55$ since $5^{55}<2^{128}$.
We have that the largest 32-bit floating-point number is $\approx 3.4 \times 10^{38}$ so that exponents  
greater than $55$ are irrelevant in the sense that they result in infinite values as long as the significand is non-zero. However, 64-bit floating-point numbers can be larger.
For larger values of $q$, we have that the most significant 64~bits of the truncated product are exact when  the second most significant word is not filled with 1-bits ($2^{64}-1$). When it is filled with ones then  the computation of a third (or subsequent) multiplication could add one to this second word which would overflow into the most significant word, adding a value of one again. It could maybe result into a value that was seemingly just slightly under the midpoint between two floating-point values (and would thus be rounded down) to appear to switch to just over the  midpoint between two floating-point values (and to be thus rounded up).
In this unlikely case, when the least significant 64~bits  of the most significant 128~bits of the computed product are all ones, we choose  to fall back on a higher-precision approach. 

\begin{techremark}
To be more precise, before we fall back, we could check that out of the most significant 128~bits, 
all but the leading  55~bits (64-bit case) or leading 26~bits are ones, instead of merely checking the least significant 64~bits. However, we can show
that checking that  the least significant 64~bits of the truncated 128-bit product are all ones is sufficient. We sketch the proof.
Before the computation of the second product, the most
significant 64 bits of the product have trailing ones (e.g., 9~bits for 64-bit numbers), otherwise we would not have needed a second multiplication.
If the second multiplication does not overflow into the most  significant 64~bits, then the final result still has trailing ones in its  most  significant 64~bits. Suppose it is not the case: there was an overflow in the most significant 64~bits following the second multiplication.
 When the second product overflows into the most significant 64~bits (turning these ones into zeros), the second 64-bit word of the product is at most $2 \times (2^{64}-1) \bmod 2^{64} = 2^{64}-2$. Thus when the second most significant 64~bits are all ones then there was no overflow into the most significant 64-bit word by the second multiplication and the least significant bits of the first 64-bit word of the product are also ones.
\end{techremark}
\end{itemize}

Once we have determined that we have sufficient accuracy,  we either need to check
for round-to-even cases ($q\leq 23$), see \S~\ref{sec:roundevenpositive}, or else we proceed
with the general case. 
In the general case, we consider  the most significant 54~bits (64-bit numbers) or 24~bits (32-bit numbers), not counting the possible leading zero, and then we round up or down based on the least significant bit.  E.g., we round down to 53~bits in the 64-bit cases when the least significant bit out of 54~bits is zero, otherwise we round up. If all of the bits are ones, rounding up overflows into a more significant bit and we shift  the result by one bit so we have the desired number of bits (53~bits or 24~bits). 
In \S~\ref{sec:expo}, we show how to compute the binary exponent efficiently.
Whenever we find that
the resulting binary floating-point number $m \times 2^p$ is too large, we return the infinite value. The upper limits are $ 2^{1024}$ ($\approx 1.7976934\times 10 ^{308}$) for 64-bit numbers  and $2^{128}$ ($3.4028236\times 10 ^{38}$) for 32-bit numbers. The infinite values are signed ($+\infty, - \infty$).

\begin{numberedexample}
Consider the string 
\texttt{2440254496e57}. We get $w=2440254496$ and $q=57$. We normalize the decimal significand $w$ by shifting it by 32~bits
so that the most significant bit is set (considering $w$ as a 64-bit word).
We look up the 64~most significant bits of $5^{57}$ which are given by the word value 11754943508222875079 (or $5^{57}\div 2^{133-64}$).
The most significant 64~bits of the product are, in hexadecimal notation, \texttt{0x5cafb867790ea3ff}. All of the least significant 9~bits are set. Thus we need to execute a second multiplication. In practice, it is a rare instance. We load the next most significant 64~bits of the product (\texttt{0xaff72d52192b6a0d}) and compute the second product. After updating our product, we get that the most significant word is  \texttt{0x5cafb867790ea400} whereas the second most significant word is \texttt{0x1974b67f5e78017}. The most significant 55~bits of the product have been computed exactly. 
The most significant bit of  \texttt{0x5cafb867790ea400} is zero. We select the most significant 55~bits of the word: 13044452201105234. It is an even value so we do not round up. Finally, we shift the result to get the binary significand $m=6522226100552617\in [2^{52}, 2^{53}) $.
\end{numberedexample}

\subsection{Round Ties to Even with Positive Powers}
\label{sec:roundevenpositive}

To handle rounding ties to even when the decimal power is positive ($q\geq 0$), we only have to be concerned when  the exponent $q$ is sufficient small: $q\leq 23$ assuming that we want to support both 64-bit  numbers (which require $q\leq 23 $) and  32-bit numbers (which require $q\leq 10 $) due to the fact that we require the decimal significand to be smaller than $2^{64}$ (see \S~\ref{sec:roundtoeven}). In such cases, the  product corresponding to $m \times 5^q$  is always exact after one multiplication (since $5^q<2^{64}$).

We need to round-to-even when the resulting binary significand uses one extra significant bit for a total of 54~bits in the 64-bit case and
25~bits in the 32-bit case. We can check for the occurrence of a round-to-even case as follows:
\begin{itemize}
\item The most significant 64-bit word has exactly 10~trailing zeros (64-bit case) or exactly 39~trailing zeros (32-bit case).\footnote{We assume that the most significant word of the product contains a leading 1-bit, maybe following a shift. If the word has a leading zero,  we must subtract one from these numbers: 9 for the 64-bit case and 38 for the 32-bit case.}
\item The second 64-bit word of the product containing the least significant bits equals  zero. 
\end{itemize}
In such cases, we get a 54-bit  (64-bit case) or a 25-bit (32-bit case) word that ends with a 1-bit. We need to round to  a 53-bit or 24-bit value. We can either round up or round down. To implement the round to even, we check the second last significant bit. If it is zero,  we round down; if it is one, we round up. When rounding up, if we overflow because we only have 1-bits,  we shift  the result to get the desired number of bits (53~bits or 24~bits).

When the round-to-even case is not encountered,  we either round up or down as in the general case. Given that we have computed the exact product ($q\leq 23$), there can be no ambiguity. 

\begin{techremark}
When $q\leq 23$, the second most significant 64-bit word  must have some trailing zeros, whether we are in a round-to-even case or not. Indeed, because we always have $5^q \leq 2^{54}$
and because we store the powers of five shifted  so that they have a leading 1-bit in their respective 64-bit words (e.g., $5^{23}$ is loaded as $5^{23} \times 2^{10}$), we have at least 10~trailing zeros, irrespective of the word $w$ as long as $q\leq 23$. Hence we can replace a
comparison with zero with a check that it is no larger than a small value like 1, if it is convenient. It turns out to be useful to us for simplifying Algorithm~\ref{algo:fancytotalalgo}, see line~\ref{line:smallerthanone}.
\end{techremark}

\begin{numberedexample}
Consider the input string \texttt{9007199254740993}. It is equal to $2^{53}+1$, a number that cannot be represented as a standard 64-bit floating-point number.
After parsing the string, we get $w= 9007199254740993$ and $q= 0$.
We normalize $w$ to 9223372036854776832 ($2^{63} + 2^{10}$). We multiply
this shifted $w$ by $5^q$ which, in this case, is just 1 ($5^q = 1$). The precomputed $5^q$ has been normalized 
to the 64-bit value $2^{63}$. Hence we get the product $2^{63} w$.
We then select the most significant 54~bits of the result after skipping the leading zero~bit ($2^{53}+1$). Because this result has its least significant bit set, and because all the bits that were not selected are zero, we know that we need to  round to even. The second least significant bit is zero, and we know that we need to round down. We therefore end up with a significand of one and an exponent of $2^{53}$.
\end{numberedexample}





\section{Division by Powers of Five}
\label{sec:divisionpowerfive}



We turn our attention to negative decimal exponents ($q<0$).
Consider the decimal number
$9.109\times 10^{-31}$ corresponding to the mass of the electron. We can write it as
$9109\times 10^{-34}$ or 
$9109\times 2^{-34} \times 5^{-34}$. Though we can represent the integer
$5^{34}$ using few binary words, the value $5^{-34}$ could only be approximated
in binary form. E.g., it is approximately equal to 
$4676805239458889  \times 2 ^ {-131}$.
To convert it to a binary floating-point number, we must find $m$ and $p$ such that
$9109\times 2^{-34} \times 5^{-34} \approx m \times 2^p$.

Consider the general case, replacing $9109\times 2^{-34}$ by $w \times 10^q$ where $0<w< 2^{64}$ and $q<0$.
We need to approximate the decimal floating-point number $w \times 10^q$ as closely 
as possible with the binary floating-point number $m\times 2^p$.
We formalize $m\times 2^p \approx w \times 10^q$ as the equation $m\times 2^p + \epsilon 2^p = w \times 10^q$ where $ \epsilon$ is the approximation error. If the binary significand $m$
is chosen optimally, then the error must be as small as possible: $\vert \epsilon \vert \leq 1/2$. 
Dividing throughout, we get 
$m +  \epsilon =  w \times 2^{q-p} / 5^{-q} $
or $m = \round ( w \times 2^{q-p} / 5^{-q})$.
Thus the problem is reduced to a division by an integer power of five ($5^{-q}$) of an integer $w$
that fits in a 64-bit word  multiplied by a power of two.
 The correct value of $p$ is such that $m$ should fit in 53~bits (64-bit case) or 24~bits (32-bit case). In practice,  we can derive the correct value $m$ after rounding if we compute the division $w \times  2^{b}/ 5^{-q}$ for a sufficiently large power of two $2^b$.
 
 Dividing a large integer by another large integer could be expensive.
Thankfully, we can compute the quotient and remainder of  such a division by the divisor $d= 5^{-q}$ using a multiplication followed
by a right shift when the divisor $d$ is known ahead of time. Many optimizing
compilers have been using such a strategy for decades. We apply the following result derived from Warren~\cite{lemire2019faster,warr:hackers-delight-2e}.

\begin{theorem}\label{thm:remainder}
Consider an integer divisor $d>0$ and a range of  integer numerators $n\in[0, N]$ where $N\geq d$ is an integer.
We have that  
\begin{eqnarray*}
n\div d = \floor(c \times n/t)
\end{eqnarray*}
 for all integer numerators $n$ in the range if and only if
\begin{eqnarray*}1/d \leq c/t< \left (1+\frac{1}{N -  \remainder(N+1,d)} \right )1/d.\end{eqnarray*}
\end{theorem}

Intuitively,  if $c/t$ is close to $1/d$,  we can
replace $n/d$ by $n \times c / t$.
We apply Theorem~\ref{thm:remainder} by picking a power of two for $t$ so that the computation of
$ \floor(c \times n/t)$ can be implemented as a multiplication followed by a shift. Given $t$, the smallest constant $c$ such that $1/d \leq c/t$ holds is $c = \ceiling (t/d)$.  It remains
 to check that $c=\ceiling (t/d)$ is sufficiently close to $t/d$:
\begin{eqnarray*}\ceiling (t/d)\times  d <\left (1+\frac{1}{N -  \remainder(N+1,d)} \right )t.\end{eqnarray*}
We have that $\ceiling (t/d) \times d \leq t -1 + d$.
By substitution, we get the necessary condition $d-1 < t/ (N -  \remainder(N+1,d)) \leq t/N$. And so we have the convenient sufficient
condition $t > (d-1) \times N$.
We have shown Corollary~\ref{cor:remainder}.

\begin{corollary}\label{cor:remainder}
Consider an integer divisor $d>0$ and a range of  integer numerators $n\in[0, N]$ where $N\geq d$ is an integer.
We have that  
\begin{eqnarray*}
n\div d = \floor \left (\ceiling \left  (\frac{t}{d} \right ) \frac{n}{t} \right )
\end{eqnarray*}
 for all integer numerators $n$ if  $t > (d-1) \times N$.
\end{corollary}


Observe how we multiply the numerator by a reciprocal ($\frac{t}{d}$) that is rounded up. It suggests that we need to store precomputed reciprocals while rounding up. In contrast, when processing positive decimal exponents, we would merely truncate the power, thus effectively rounding it down.
    
We consider two scenarios. Firstly, the case with  negative powers with exponents near zero  requires more care because we may need to round to the nearest even (\S~\ref{sec:nearone}). Secondly, we consider the general case ($5^{-q}\geq 2^{64}$)  (\S~\ref{sec:otherneg}).


 

\subsection{Negative Exponents Near Zero ($q \geq -27$)}

\label{sec:nearone}





We are interested in identifying  cases when  rounding to even is needed because of a tie (e.g., $2^{-1022} + 2^{-1074}+ 2^{-1075}$ is one such number). Round-to-even  only happens for negative values
of $q$ when $q\geq -4$ in the 64-bit case and when  $q\geq -17$ in the 32-bit case (see \S~\ref{sec:roundtoeven}). In either cases, we have
that $5^{-q}$ fits in a 64-bit word (i.e., $5^{-q} < 2^{64}$).

Theorem~\ref{thm:remainder} tells us how to quickly compute a quotient given a precomputed
reciprocal, but for values of $q$ near zero, we need to be able to detect when the remainder is zero, so we know when to round ties to even.
 We use the following technical lemma which allows us to verify quickly whether our significand is divisible by the power of five.
 \begin{lemma} \label{lemma:tech1}Given an integer divisor $d>0$, pick an integer $K>0$ such that $2^K\geq d$. Given an integer $w>0$, then 
 $(w \times 2^{K}) \div d $ is divisible by $2^{K}$ if and only if $w$ is divisible by $d$.
 \end{lemma}
 \begin{proof}If $w$ is divisible by $d$ then   $(w \times 2^{K}) \div d = (w \div d) \times 2^{K}$ and thus  $(w \times 2^{K}) \div d $ is divisible by $2^{K}$. Suppose that $(w \times 2^{K}) \div d $ is divisible by $2^{K}$, then we
 can write $w \times 2^{K} = d \times 2^{K} \times \gamma + \rho$ where $\gamma , \rho$ are non-negative integers with
 $\rho < d$. But because $\rho < 2^K$, we must have $\rho= 0$ and so 
  $w \times 2^{K} = d \times 2^{K} \times \gamma $ or  $w  = d \times \gamma $ which shows that $w$ is divisible by $d$.
 \end{proof}

 To detect a zero remainder, we can consider the decimal 
significand $w$ as a 128-bit integer (made of two 64-bit words). 
Within the most significant 64-bit word, we shift $w$ by an adequate power of two so that the 63$^\textrm{rd}$~bit has value 1: let us call the result $w'$.
The second, least significant, word is virtual and assumed to be zero. In effect, we consider the 128-bit  value $2^{64}\times w'$. We are going to treat mathematically this 128-bit value as a 127-bit value equal to $w' \times 2^{63}$, ignoring the least significant bit. We could ignore more than one bit, but it is convenient to ignore the least significant bit.


As long as $5^{-q} < 2^{63}$, we have that
 $(w' \times 2^{63})\div 5^{-q} \geq 2^{63}$ and we have more than enough
 accuracy to compute the desired binary significand. Hence, we can
 go as low as $q=-27$.\footnote{Though we go as low as $q=-27$, values of $q$ smaller than $-17$ cannot lead to a tie-to-even scenario using 64-bit or 32-bit floating-point numbers.}
We apply Theorem~\ref{thm:remainder} to divide this value  ($w' \times 2^{63}$) no larger than $N=2^{127}$ by the power of five $5^{-q}$:
we set $t = 2^{b}$ with $b= 127 + \ceiling(\log_2(5^{-q}))$
and the reciprocal $c=  \ceiling (t/d) = \ceiling (2^{b}/5^{-q})$. 
We can check that $c\in[2^{127},2^{128})$.
We can verify the condition ($ t> (d -1)\times N$)
of  Corollary~\ref{cor:remainder}. 
 The values of the 128-bit reciprocals are given in Table~\ref{tab:tabnearone}.
\begin{table}[!tbh]
    \centering
     \rowcolors{2}{gray!25}{white}
    \begin{tabular}{ccc}
    \toprule
$q$ & reciprocal $\div 2^{64}$ & reciprocal $\bmod 2^{64}$  \\
    \midrule     
-1 & \texttt{cccccccccccccccc} & \texttt{cccccccccccccccd} \\
-2 & \texttt{a3d70a3d70a3d70a} & \texttt{3d70a3d70a3d70a4} \\
-3 & \texttt{83126e978d4fdf3b} & \texttt{645a1cac083126ea} \\
-4 & \texttt{d1b71758e219652b} & \texttt{d3c36113404ea4a9} \\
-5 & \texttt{a7c5ac471b478423} & \texttt{0fcf80dc33721d54} \\
-6 & \texttt{8637bd05af6c69b5} & \texttt{a63f9a49c2c1b110} \\
-7 & \texttt{d6bf94d5e57a42bc} & \texttt{3d32907604691b4d} \\
-8 & \texttt{abcc77118461cefc} & \texttt{fdc20d2b36ba7c3e} \\
-9 & \texttt{89705f4136b4a597} & \texttt{31680a88f8953031} \\
-10 & \texttt{dbe6fecebdedd5be} & \texttt{b573440e5a884d1c} \\
-11 & \texttt{afebff0bcb24aafe} & \texttt{f78f69a51539d749} \\
-12 & \texttt{8cbccc096f5088cb} & \texttt{f93f87b7442e45d4} \\
-13 & \texttt{e12e13424bb40e13} & \texttt{2865a5f206b06fba} \\
-14 & \texttt{b424dc35095cd80f} & \texttt{538484c19ef38c95} \\
-15 & \texttt{901d7cf73ab0acd9} & \texttt{0f9d37014bf60a11} \\
-16 & \texttt{e69594bec44de15b} & \texttt{4c2ebe687989a9b4} \\
-17 & \texttt{b877aa3236a4b449} & \texttt{09befeb9fad487c3} \\
\midrule
-18 & \texttt{9392ee8e921d5d07} & \texttt{3aff322e62439fd0} \\
-19 & \texttt{ec1e4a7db69561a5} & \texttt{2b31e9e3d06c32e6} \\
-20 & \texttt{bce5086492111aea} & \texttt{88f4bb1ca6bcf585} \\
-21 & \texttt{971da05074da7bee} & \texttt{d3f6fc16ebca5e04} \\
-22 & \texttt{f1c90080baf72cb1} & \texttt{5324c68b12dd6339} \\
-23 & \texttt{c16d9a0095928a27} & \texttt{75b7053c0f178294} \\
-24 & \texttt{9abe14cd44753b52} & \texttt{c4926a9672793543} \\
-25 & \texttt{f79687aed3eec551} & \texttt{3a83ddbd83f52205} \\
-26 & \texttt{c612062576589dda} & \texttt{95364afe032a819e} \\
-27 & \texttt{9e74d1b791e07e48} & \texttt{775ea264cf55347e} \\
    \bottomrule
    \end{tabular}
    \caption{Values of the 128-bit reciprocals in hexadecimal form for powers with negative exponents near zero as two 64-bit words. The reciprocal is given by $\ceiling(\frac{2^{b}}{5^{-q}})$ with $b=127+\ceiling(\log_2 (5^{-q}))$.}
    \label{tab:tabnearone}
\end{table}
By Theorem~\ref{thm:remainder}, we have that
$(2^{63}\times w') \div  5^{-q} = ((2^{63}\times w')   \times c) \div 2^{b}$. The 128-bit constant $c$ is precomputed for all relevant powers
of $q= -1, -2,  \ldots, -27$. 

Computationally, we do not need to store the 
 product of the shifted 128-bit $w$ with the 
128-bit reciprocal $c$ as a 256-bit product: we compute only the first two words (i.e., the most significant
128-bit). And then we select only the most significant 127~bits.  

Thus checking that the second (least significant) word of the computed product is zero except maybe for the least significant bit is enough to determine that our word ($w$) is divisible by five. By dividing (the shifted) $w$ (that is $w'$) by $5^{-q}$, we effectively compute a binary significand $2m+1$ as per the equation
 $(2m+1) \times 2^{p-1} = w \times 10^q$.
 We need to round to even whenever $(2m+1) \in [2^{53}, 2^{54})$ for 64-bit floating-point numbers and whenever $(2m+1) \in [2^{24}, 2^{53})$ for 32-bit floating-point numbers. Otherwise, we round  the result normally to the nearest 53-bit word (64-bit numbers) or 24-bit word (32-bit numbers).

We need  up to two 64-bit multiplications. We stop after the first multiplication if and only if  the least significant bits of the most significant words are not all ones. We want that the most significant bits of the most significant 64-bit word are exact: 55~bits for 64-bit numbers and 26~bits for 32-bit numbers. We have that $64-55$ is 9 and $64-26$ is 38. Hence, 
we  check the least significant  9~bits for 64-bit numbers and the least significant  38~bits for 32-bit numbers.  
As long as we are not in a round-to-even case, we
round up or down based on the least significant selected bit. If all bits have the value 1, then rounding up overflows into a more significant bit and we must shift  by one bit.

We need to ensure that we correctly identify all  ties requiring the round-to-even strategy. Specifically, we need to never incorrectly classify a number as a tie, and we need
to never miss a tie.
\begin{itemize}
\item We need to be concerned about a false round-to-even scenario when, after stopping with just one multiplication, we end up with a misleading result that could pass as a round-to-even case. Indeed, we can stop
after one multiplication when  the least significant bits of the most significant words are all zeros. 
However, a round-to-even case cannot occur after a single multiplication:
\begin{enumerate}
\item It could happen if the least significant $64+9$~bits of the product are zeros.  The 128-bit product of two 64-bit words may only have as many trailing zeros as the sum of the number of trailing zeros of the first 64-bit word with the number of trailing zeros of the second word. To get a total of $64+9$~trailing zeros, assuming that both words are non-zero, we have the necessary conditions that both words must  have at least 10~trailing zeros. Thus, for this problem to occur, we need for the most significant 64-bit word of the reciprocal $c$ to have at least 10~trailing zeros. We can check that it does not happen: there are only 17~powers to examine. We find at most 2~trailing zeros. See Table~\ref{tab:insane} (third column).
\item A false round-to-even may also happen if  all the least significant $64+9$~bits of the product are zeros, except for the least significant bit. However, for the 128-bit product of two 64-bit words to have its least significant bit be 1, we need for both of the 64-bit words to have their least significant bits set to 1 (they are odd). Given an odd 64-bit integer, there is  only one other 64-bit integer such as the least significant 64~bits of the product is 1. Indeed suppose that $a \times b_1 = 1 \bmod 2^{64}$ and
$a \times b_2 = 1 \bmod 2^{64}$ for numbers in $[0,2^{64})$ then $a \times (b_1-b_2) = 0 \bmod 2^{64}$ which implies that $b_1=b_2$.
They are effectively multiplicative inverses (modulo $2^{64}$). We can thus compute the multiplicative inverses (see Fig.~\ref{fig:multiinverse}) and check the full 128-bit product. Again, we only need to examine 17~powers. We find the powers that have an odd integer in their most significant 64~bits, we compute the multiplicative inverse  and we compute the full product. Looking at the most significant 64~bits of the resulting product, we find that they have at most 5~trailing zeros. See Table~\ref{tab:insane} (last column).
\end{enumerate}

\begin{table}[tbh]
    \centering
    \begin{tabular}{ccccc}
    \toprule
$q$ & reciprocal $\div 2^{64}$  & inverse &  product & 0s \\
    \midrule     
-3 &  \texttt{83126e978d4fdf3b}  & \texttt{c687d6343eb1a1f3} & \texttt{65a5cdedb181dc22} & 1\\
-4 &  \texttt{d1b71758e219652b}  & \texttt{6978533007ec3183} & \texttt{5666aa8c1bca175b} & 0\\
-5 &  \texttt{a7c5ac471b478423}  & \texttt{b464ceec1a874b8b} & \texttt{76390df51733b898} & 3\\
-6 &  \texttt{8637bd05af6c69b5}  & \texttt{2d28ff519dc1fc9d} & \texttt{17ad4acbd85ad372} & 1\\
-9 &  \texttt{89705f4136b4a597}  & \texttt{47a5ffb53d302a27} & \texttt{26774920b7634d5b} & 0\\
-12 &  \texttt{8cbccc096f5088cb}  & \texttt{ccda17e7d0519ce3} & \texttt{709e5881abf430de} & 1\\
-13 &  \texttt{e12e13424bb40e13}  & \texttt{a976a8f009f3ec1b} & \texttt{950fca8d051f7f36} & 1\\
-14 &  \texttt{b424dc35095cd80f}  & \texttt{4776114e932f16ef} & \texttt{32494e3df377fbda} & 1\\
-15 &  \texttt{901d7cf73ab0acd9}  & \texttt{b7d434f9093d1369} & \texttt{677c8a9266f5159b} & 0\\
-16 &  \texttt{e69594bec44de15b}  & \texttt{30fad280461f66d3} & \texttt{2c1df79145125a20} & 5\\
-17 &  \texttt{b877aa3236a4b449}  & \texttt{89ee897ef59d7df9} & \texttt{6363ec689fe3979b} & 0\\
    \bottomrule
    \end{tabular}
    \caption{Values of the most significant 64~bits of the 128-bit reciprocals 
    in hexadecimal form for  powers of negative exponents near zero 
    and the multiplicative inverse modulo $2^{64}$ of the reciprocal for odd reciprocals 
    ($q=-1, -2, -7, -8, -10, -11$ are omitted since their reciprocals are even). We compute the most significant bits 
    of the 128-bit product between the reciprocal and its inverse. We indicate the number of trailing zeros for  the most significant bits of the product.}
    \label{tab:insane}
\end{table}

\item We need to be concerned with the reverse scenario where, after a single multiplication, we stop the computation and fail to detect an actual round-to-even case. 
If we stop after one multiplication, then at least one of the least significant bits (9~bits for 64-bit numbers, 38~bits for 32-bit numbers) of the most significant 64~bits is zero.
In such a case, the 128~most significant bits of the full (exact) product must end with a long stream of zeros, except maybe for the least significant bit. We know that the most significant 64~bits are exact after a single product, except maybe for the need to increment by 1.  
The most significant 64~bits cannot be exact after one multiplication if we have a round-to-even case. So we must increment them
by 1 following the second multiplication, and then the final result contains at least one non-zero bit in the least significant bits (9~bits for 64-bit numbers, 38~bits for 32-bit numbers) of the most significant 64~bits. It contradicts the fact that we had an actual round-to-even case.
Hence, we cannot fail to detect an actual round-to-even case by stopping the computation after one multiplication.
\end{itemize}

Thus we can identify accurately the round-to-even cases. In these cases, we proceed as in \S~\ref{sec:roundevenpositive}.  After discarding a potential leading zero-bit, we have 54~bits (64-bit case) and 25~bits (32-bit case). The least significant bit is always a 1-bit. We round down when the second least significant bit is zero, otherwise we round up. When rounding up, we might overflow into an additional bit if we only have ones, in such a case we shift  the result.

\subsection{Other Negative Powers ($q<-27$)}

\label{sec:otherneg}
Consider the case where the decimal exponent is far  from zero ($q< -27$). In such cases, the decimal number can never be exactly in-between two floating-point numbers: thus with a single extra bit of accuracy, we can safely either round up or down.

  The smallest positive value that can be represented using a 64-bit floating-point number is $2^{-1074}$. For 32-bit numbers, we have the larger value $2^{-149}$. Because  we have that $w\times 10 ^{-343} < 2^{-1074}$ for all $w<2^{64}$, it follows that we never have to be concerned with overly small decimal exponents: when $q<-342$, then the number is assuredly zero.

From the decimal number $m \times 10^q$, we seek the  binary significand  $m = \round ( w \times 2^{q-p} / 5^{-q})$ where the binary power $p$ is chosen such that $m$ is within the range of the floating-point numbers (e.g., $m\in [2^{52}, 2^{53})$).
It is enough to compute $m'=\floor(w \times 2^b/ 5^{-q}) $ with $b$ large enough that
$m' \geq 2^{53}$ so that we can compute $m = \round ( w \times 2^{q-p} / 5^{-q})$ accurately by selecting the most significant 53~bits (64-bit numbers) or 24~bits (32-bit numbers) of the wider value $m'$ and then round it up (or down) based on the $54^{\mathrm{th}}$ or $25^{\mathrm{th}}$ bit value.

 We can pick $b=64+ \ceiling( \log_2 5^{-q})$. 
We apply  Corollary~\ref{cor:remainder} with  $t = 2^{2b}$, $d= 5^{-q}$,
and $N= (2^{64} - 1) 2^b$. 
We precompute $c=  \ceiling (t/d) = \ceiling (2^{2 b}/5^{-q})$ for all relevant powers of $q\geq-342$. 
See Table~\ref{tab:awayfromone}.
We  only store  the most significant 128~bits of $c$, and rely on a truncated multiplication.
Because there is no concern with rounding to even, we can safely round up from the most significant bits of the computed quotient. 
We do just one multiplication if  it provides
the number of significant bits of the floating-point standard (53~bits for 64-bit numbers and 24~bits for 32-bit numbers) plus one additional bit to determine the rounding direction, and yet one more bit to handle the scenario where the computed product has a leading zero. We  always stop after this second multiplication when we have a truncated product with the second most significant word not filled with ones ($2^{64}-1$).
Otherwise, we fall back on a higher-precision approach, an unlikely event.
 
After possibly omitting the leading zero of the resulting product, we select the most significant bits (54~bits in the 64-bit case, 25~bits in the 32-bit case). We then round up or down based on the least significant bit to 53~bits (64-bit case) or to 24~bits (32-bit case). When rounding up, we might overflow to an additional bit if we have all ones: in such case we shift  to get back 53~bits (64-bit case) or 24~bits (32-bit case).
 
\begin{table}[tbh]
    \centering
     \rowcolors{2}{gray!25}{white}
    \begin{tabular}{ccc}
    \toprule
$q$ & reciprocal (64~msb) & reciprocal (next 64~msb)  \\
    \midrule     
 -40 & \texttt{8b61313bbabce2c6} & \texttt{2323ac4b3b3da015} \\
-39 & \texttt{ae397d8aa96c1b77} & \texttt{abec975e0a0d081a} \\
-38 & \texttt{d9c7dced53c72255} & \texttt{96e7bd358c904a21} \\
-37 & \texttt{881cea14545c7575} & \texttt{7e50d64177da2e54} \\
-36 & \texttt{aa242499697392d2} & \texttt{dde50bd1d5d0b9e9} \\
-35 & \texttt{d4ad2dbfc3d07787} & \texttt{955e4ec64b44e864} \\
-34 & \texttt{84ec3c97da624ab4} & \texttt{bd5af13bef0b113e} \\
-33 & \texttt{a6274bbdd0fadd61} & \texttt{ecb1ad8aeacdd58e} \\
-32 & \texttt{cfb11ead453994ba} & \texttt{67de18eda5814af2} \\
-31 & \texttt{81ceb32c4b43fcf4} & \texttt{80eacf948770ced7} \\
-30 & \texttt{a2425ff75e14fc31} & \texttt{a1258379a94d028d} \\
-29 & \texttt{cad2f7f5359a3b3e} & \texttt{96ee45813a04330} \\
-28 & \texttt{fd87b5f28300ca0d} & \texttt{8bca9d6e188853fc} \\
    \bottomrule
    \end{tabular}
    \caption{Values of the 128-bit reciprocals in hexadecimal form for negative exponents as two 64-bit words. The reciprocal is given by $\ceiling(\frac{2^{2b}}{5^{-q}})$ with $b=64+ \ceiling( \log_2 5^{-q})$.}
    \label{tab:awayfromone}
\end{table}
 
\begin{numberedexample}\label{remark:funcase}
Consider the case of the string \texttt{9.109e-31}. We parse it as $9109 \times 10^{-34}$.
We load up the most significant 64~bits of the reciprocal corresponding to $q=-34$ which is
\texttt{0x84ec3c97da624ab4} in hexadecimal form (see Table~\ref{tab:tabnearone}).
We normalize 9109 so that, as a 64-bit word, its most significant bit is 1: $9109\times 2^{50}$.
We multiply the two words to get that the most significant 64~bits of the
product are \texttt{49e6a7201cf62db0} whereas the next most significant 64~bits are
\texttt{0x5b10000000000000}. We stop the computation since the second word is
not filled with ones. 
The most significant bit of the product contains a 0.
We shift the most significant 64~bits by 9~bits to get 10400639386286870.
The least significant bit is zero so we round down to 5200319693143435 or
\texttt{0x1279a9c8073d8b} in hexadecimal form. We get that $9109 \times 10^{-34}$
is the floating-point number \texttt{0x1.279a9c8073d8bp-100}. See Example~\ref{example:expo} in \S~\ref{sec:expo} to learn how we determine that the binary exponent is -100.
\end{numberedexample}

\subsection{Subnormals}\label{sec:subnormals}
To represent values that are too small, the floating-point standard uses special values  called subnormals. Whenever we end up with a value  $m \times 2^{p}$ with $m\in [2^{52}, 2^{53})$ (64-bit case) or $m\in [2^{23}, 2^{24})$ (32-bit) but with $p$ too small, smaller than $-1022-52$ in the 64-bit case or
smaller than $-126-23$ in the 32-bit case, we fall back on the subnormal representation. It uses a small value for the exponent to represent values in the range $[2^{-1022-52}, 2^{-1022})$ (64-bit case) or in the range $[2^{-126-23}, 2^{-126})$ (32-bit case). The 
values are given by $m \times 2^{-1022-52}$ (64-bit) or  $m \times 2^{-126-23}$ (32-bit) while allowing $m$ to be any positive value no larger than $2^{52}$ or $2^{23}$. 

To construct the subnormal value, we take the original binary significand $m$ and we divide it by a power of two, with rounding. Thus, for example, if we are given the 64-bit value
 $(2^{53}-1) \times 2^{-1022-54}$, we observe that the power of two is too small ($-1022-53<-1022-52$) by exactly two. Thus we take the binary significand $ 2^{53}-1$
 and divide it by four, with rounding: we get $2^{51}$ and so we get the subnormal
 floating-point number $2^{51} \times  2^{-1022-52}$.
 Thankfully, rounding is relatively easy since we never need to handle the round-to-even case with subnormals, because it only occurs with  powers of exponents near zero.

We should be mindful that, in exceptional cases, the rounding process can lead us to find that we do not have a subnormal.
Indeed, consider the value
 $(2^{53}-1) \times 2^{-1022-53}$, its power of two is too small ($-1022-53<-1022-52$) by exactly one. We take the binary significand $ 2^{53}-1$
 and divide it by two, with rounding, getting $2^{52}$ and so we end up with the normal number $2^{52} \times 2^{-1022-52}$. 

 \section{Computing the Binary Exponent Efficiently}
\label{sec:expo}

We are approximating a decimal floating-point number $w \times 10^q $ with a binary 
floating-point number $m \times 2^p$. 
We must compute the binary exponent $p$.
Starting from the power of ten $10^q$, we want write it as a value in $[1,2)$, as prescribed by the floating-point standard, multiplied by a power of two. We have two distinct cases depending on the sign of $q$:
\begin{itemize}
\item when $q\geq 0$, we have    $10^q = 2^q \times 5^q =  \frac{5^q}{2^{\floor(\log_2 5^q)}}  \times 2^{q+ \floor(\log_2 5^q)}$,
\item when $q< 0$, we have   $10^q = 2^q \times 5^q = \frac{ 2^{\ceiling(\log_2 5^{-q})}}{ 5^{-q}}  \times 2^{q - \ceiling(\log_2 5^{-q})}$.
\end{itemize}
We can verify that both constraints are satisfied: $ 5^q/2^{\floor(\log_2 5^q)} \in [1,2)$ and
$2^{\ceiling(\log_2 5^{-q})}/ 5^{-q} \in [1,2)$. Hence we have that the binary powers
corresponding to the powers of ten are given by  $q + \floor (\log_2( 5^q )) = q - \ceiling(\log_2 5^{-q})$.
For example, we have that $10^5 = 5^5/2^{11 }\times 2^{16} =  1.52587890625 \times 2^{16}$ since  $ \floor( \log_2 5^5) = 11$.
Computing $q+\log_2( 5^q )$ could require an expensive iterative process. 
The decimal exponent $q$ is in  limited range of values, say $q \in (-400,350)$.
We have that $q+\log_2( 5^q )= q + q \log_2( 5) = q (1 + \log_2( 5))$ and $1 + \log_2( 5) \approx 217706 /2^{16}$. We can check that over the interval 
$q \in (-400,350)$, we have that $q + \floor (\log_2( 5^q )) = (217706 \times q) \div 2^{16} $ (exactly)  as one can verify numerically. 
 The division (by $2^{16}$) can be implemented as a logical shift. Thus we only require 
 a multiplication followed by a shift.
 We initially derived this efficient formula using a satisfiability-modulo-theories (SMT) solver~\cite{dutertre2014yices}.

In our algorithm, we normalize the decimal significand so that it is in $[2^{63}, 2^{64})$.  That is, given the string \texttt{1e12}, we first parse it as the decimal significand $w=1$ and the decimal exponent $q=12$. We then normalize $w=1$ to $w'=2^{63}$ (shifting it by 63~bits) and we proceed with the computation of the binary significand.
Had we started with $w=2^4$ (say), then we would have shifted by only $63-4$~bits and then the binary exponent must be incremented by 4. For example, using the input string \texttt{16e12} instead of \texttt{1e12}, we would have used the decimal significand $w= 16$ but still ended up with the normalized significand $w'=2^{63}$. Yet the binary exponent of 
\texttt{16e12} is clearly 4 more than the binary exponent of \texttt{1e12}. In other words, we need to take into account the number of leading zeroes of the decimal significand.
Thus we increment the binary exponent by $63-l$ where $l$ is the number of leading zeros of the original decimal significand $w$ as a 64-bit word.

For powers of ten, the product of the normalized significand with either the power of five or its reciprocal  has a leading zero since $2^{63} \times (2^{64}-1) < 2^{127} $. When the product is larger and  it overflows in the most significant bit, then the binary exponent must be 
incremented by one.
Thus we finally have the following formula
\begin{equation*}
 \left ( \left (217706 \times q\right ) \div 2^{16} \right) + 63 - l + u
\end{equation*}
where $u $ is the value of the most significant bit of the product (0 or 1) and where $l$ is the number of 
leading zeros of $w$. 
 
Furthermore, when we round up the resulting significand, it may sometimes overflow: e.g., if the most significant bits of the product are
all ones,  we overflow to a more significant bit and we need to shift  the result. In such cases, we increment the binary exponent by one.



 When serializing the exponent in the IEEE binary format, we need to add either 1023 (64-bit) or 127 (32-bit) to the exponent; these constants (1023 and 127) are sometimes called  \emph{exponent biases}. For example, the 64-bit binary exponent value from \num{-1022} to \num{1023} are stored as the unsigned integer values from \num{1} to \num{2046}. The serialized exponent value 0 is reserved for subnormal values while the serialized exponent value \num{2047} is reserved for non-finite values.

\begin{numberedexample}\label{example:expo} Consider again Example~\ref{remark:funcase}. We start from  $9109 \times 10^{-34}$.
Because $q<0$, we compute  $q - \ceiling(\log_2 5^{-q})$ and get -113.
We have that 9109 has 50~leading zeros as a 64-bit word and we normalize it as $9109\times 2^{50}$.
Thus we have $I=50$ and so we need to increment the binary exponent by $63-I$ or 13. We get a binary exponent of -100. We verify that the product has a leading zero bit so we have that the binary exponent must be -100.
\end{numberedexample}

 \section{Processing Long Numbers Quickly}
\label{sec:long}

In some uncommon instances, we may have a decimal significand that exceeds 19~digits.
Unfortunately, if we are given a value with superfluous digits, we cannot  truncate the digits: it may be necessary to read tens or even hundreds of digits (up to 768~digits in the worst case). Indeed, consider the second smallest 64-bit normal floating-point value: $2^{-1022} + 2^{-1074}$ ($\approx 2.2250738585072019\times 10^{-308}$) and the next smallest value $2^{-1022} + 2^{-1073}$ ($\approx 2.2250738585072024\times 10^{-308}$). If we pick a value that is exactly in-between ($2^{-1022} + 2^{-1074}+ 2^{-1075}$),
we need to break the tie by rounding to even (to the larger value $2.2250738585072024\times 10^{-308}$ in this case). Yet any truncation of the value would be slightly closer to the lower value ($\approx 2.2250738585072019\times 10^{-308}$). We can write $2^{-1022} + 2^{-1074}+ 2^{-1075}$ exactly as a decimal floating-point value $w \times 10 ^q$ for integers $w$ and $q$, but the significand requires 768~digits. We can show that it is the worst case.

When there are too many digits, we could immediately fall back on a higher-precision approach. However,  if we just use the most significant 19~digits, and truncate any subsequent digits, we might be able to uniquely identify the exact number. It is trivially the case
if the truncated digits are all zeros, in which case we can safely dismiss the zeros.
Otherwise, if $w$ is the truncated significand, then the exact value is in the interval $(w \times 10^q,(w+1) \times 10^q)$.
Thus we may  apply our algorithm to both $w \times 10^q$ and $(w+1) \times 10^q$. If they both round to the same binary floating-point number, then this floating-point number has to match exactly the true decimal value. If $w$ is limited to 19~digits, then $w+1\leq 10^{19}<2^{64}$ so we do not have to worry about possible overflows.

To assess the effectiveness of this approach, we can try a numerical experiment. We generate random 19-digit significands  and append an exponent (e.g., \texttt{1383425612993491676e-298} and
\texttt{1383425612993491677e-298}). We find that for such randomly generated values, about 99.8\% of the successive values map to the same 64-bit floating-point number, over a range of exponents (e.g., from \num{-300} to \num{300}). We can also generate random 64-bit numbers in the unit interval $[0,1]$, serialize them to 19~digits and add one to the last digit. We get that in about 99.7\% of all cases, changing the last digit does not affect the value. In other words,  we often can determine exactly a floating-point value after truncating to 19~digits in most cases.

When it fails, we can fall back on a higher-precision approach. In our software implementation (see \S~\ref{sec:fastalgo}), we adapted a 
general implementation used as part of the Go standard library. Given that it should be rarely needed, its performance is  secondary. However,
it has to be exact.

\section{Experiments}
\label{sec:experiments}
We implemented our algorithm and published it as an  open source software library.\footnote{\url{https://github.com/fastfloat/fast_float}} It closely follows the C++17 standard for the \texttt{std::from\_chars} functions, supporting both 64-bit and 32-bit floating-point numbers. It has been thoroughly tested. Though our code is written using generally efficient C++ patterns, we have not micro-optimized it. 
Our implementation requires a C++11-compliant compiler. It does not allocate memory on the heap and it does not throw exceptions.

To implement our algorithm, we use a precomputed table of powers of five and reciprocals, see Appendix~\ref{appendix:table}. Though it  uses   \SI{10}{\kibi\byte}, we should compare it with the original Gay's implementation of \texttt{strtod} in C which uses \SI{160}{\kibi\byte}  and compiles to tens of kilobytes. Our table is used for parsing both 64-bit and 32-bit numbers.

There are many  libraries that support number parsing. For our purposes, we limit ourselves to C++ production-quality libraries. We only consider libraries that offer exact parsing. See Table~\ref{tab:test-parsers}.
We choose to omit libraries written in other programming languages (Java, D, Rust, etc.) since direct comparisons between programming
languages are error prone---see Appendix~\ref{appendix:rust_results} for benchmarks of a Rust version of our algorithm.\footnote{The release notes for Go version~1.16, which makes use of our approach, state that ``ParseFloat now uses the [new] algorithm, improving performance by up to a factor of 2.'', \url{https://golang.org/doc/go1.16}. Our C++ code was also ported to C\#, \url{https://github.com/CarlVerret/csFastFloat}, and Java, \url{https://github.com/wrandelshofer/FastDoubleParser} with good results.} 
We also include in our benchmarks the system's C function \texttt{strtod}, configured with the default locale. Though the standard Linux C++ library supports the C++17 standard, it does not yet provide an implementation of the  \texttt{std::from\_chars} functions for floating-point numbers.

To ensure reproducibility, we publish our full benchmarking software.\footnote{\url{https://github.com/lemire/simple_fastfloat_benchmark}, git tag \texttt{v0.1.0}} Our benchmarking routine takes as input a long array of strings that are parsed in sequence. Somewhat arbitrarily, we seek to compute the minimum of all encountered numbers. Such a running-minimum function carries minimal overhead compared to number parsing. Hence, we effectively measure the throughput of number parsing.
 We are also careful to use datasets containing thousands of numbers for two reasons:
 \begin{itemize}
\item On the one hand,  all measures have a small bounded error: by using  large sequence of tests, we amortize such errors. 
\item On the other hand, the performance of modern processors is often closely related to its ability to predict branches. A single mispredicted branch can waste between 10 to 20~cycles of computations. When in a repeating loop, some recent processors can learn to predict with high accuracy a few thousands of branches~\cite{seznec2011new}.
 \end{itemize}

We repeat all experiments 100~times. We avoid memory allocations throughout the process. On such a computational benchmark, timings follow a distribution resembling a log-normal distribution with a long tail associated with noise (interrupts, cache competition, context switches, etc.) and a non-zero minimum. The median is located between the minimum and the average.  Using a common convention~\cite{langdale2019parsing}, we compute both the minimum time and the average time: the difference between the 
 two is our margin of error. If the minimum time and the average time are close,  our measures are reliable. We find that the error margin is consistently less than 5\% on all platforms---often under 1\%. 
 
 On Linux platforms, we can \emph{instrument} our benchmark so that we can programmatically track the number of cycles and number of instructions retired using CPU performance counters from within our own software. Such instrumentation is precise (i.e., not the result of sampling) and does not add overhead to the execution of the code.
 Typically, the number of instructions retired by a given routine varies little from run to run and may be considered exact, especially given that we ensure a stable number of branch mispredictions.
 One benefit of instrumented code is that we can measure the effective CPU clock frequency during the benchmarked code: modern processors adjust their frequency dynamically based on load, power usage and heat. Our Linux systems are configured for performance and we observe the expected CPU frequencies.

Our benchmarks exclude disk access or memory allocation: strings are preallocated once.
To ensure a consistent  and reproducible system configuration, we run our benchmark under a privileged docker environment based on a Ubuntu 20.10 image.\footnote{\url{https://github.com/lemire/docker_programming_station}, git tag \texttt{v0.1.0}} According to our tests, the docker overhead  for purely computational tasks when the host is itself  Linux, is  negligible.
For our benchmarks, we use the GNU GCC~10.2 compiler with full optimization (\texttt{-O3 -DNDEBUG}) under Linux. Our benchmark
programs are single binaries applying the different parsing functions to the same strings.
Though we access megabytes of memory,  most of the data remains in the last-level CPU cache. We are not limited by cache or memory performance.

We rely a realistic data source that is used by the Go developers to benchmark the standard library: the canada  dataset comes from 
 a JSON file commonly used for benchmarking~\cite{langdale2019parsing}. It contains 111k  64-bit floating-point numbers serialized as strings. The canada number strings are part of geographic coordinates: e.g., \texttt{83.109421000000111}. We also include synthetic datasets containing 100k numbers each. The uniform dataset is made of 64-bit random numbers in the unit interval $[0,1]$. The integer data set is made of randomly generated 32-bit integers. Though it is inefficient to use a floating-point number parser for integer values, we believe that it might be an interesting test case. It is an instance where our code fails to show large benefits. For the synthetic dataset, we considered two subcases: the floating-point number can either be serialized using a fixed decimal significand (17~digits) or using a minimal decimal significand as 64-bit numbers (using at most 17~digits~\cite{10.1145/1806596.1806623}). We found relatively little difference in performance (no more than 10\%) on a per-float basis between these two cases. In both cases, the serialization is exact: an exact 64-bit  parser should recover exactly the original floating-point value. We present our results with the concise serialization.

\begin{table*}
\caption{\label{tab:test-parsers} Production-quality number parsing C++ libraries. Both  double-conversion and abseil have been authored by Google engineers.
}
\centering\footnotesize
\begin{minipage}{\textwidth}
\centering
\begin{tabular}{lll}\toprule
Processor   & snapshot   & link  \\ \midrule
Gay's  \texttt{strtod} (netlib) & 2001 & \texttt{www.netlib.org/fp/}   \\
 double-conversion   & version 3.1.5 & \texttt{github.com/google/double-conversion.git}   \\
abseil & 20200225.2 & \texttt{github.com/abseil/abseil-cpp} \\
\bottomrule
\end{tabular}
\end{minipage}
\end{table*}

To better assess our algorithm, we tested it on a wide range of Linux-based systems which include x64 processors, an ARM server processor and an IBM POWER9 processor. See \S~\ref{tab:test-cpus}. We report the effective frequency, that is, the CPU frequency measured during the execution of our code. Our experiments are  single-threaded: the Ampere system contains 32~ARM cores and would normally be competitive against the other systems if all cores were used. However, on a single-core basis, it is not expected to match
the other processors.

\begin{table*}[!tbh]
\caption{\label{tab:test-cpus} Systems tested 
}
\centering\footnotesize
\begin{minipage}{\textwidth}
\centering
\begin{tabular}{ccccc}\toprule
Processor    & Effective Frequency  & Microarchitecture                             & Compiler\\ \midrule
Intel i7-6700  & \SI{3.7}{\GHz} & Skylake (x64, 2015) & GCC 10.2   \\
 AMD EPYC 7262& \SI{3.39}{\GHz} & Zen~2 (x64, 2019) &  GCC  10.2 \\
Ampere & \SI{3.2}{\GHz} & ARM Skylark (aarch64, 2018) &  GCC  10.2 \\
IBM & \SI{3.77}{\GHz} & POWER9 (ppc64le, 2018) &  GCC  10.2 \\

\bottomrule
\end{tabular}
\end{minipage}
\end{table*}

We report the speed in millions of numbers per second for our different datasets and different processors in Table~\ref{tab:mfloats}. We find that the \texttt{from\_chars} function
in the  abseil library is often
superior to Gay's implementation of \texttt{strtod} (labeled as netlib) which is itself superior to both double-conversion and the  \texttt{strtod} function including the GNU standard
library. The implementation notes of the abseil library~\cite{abseil} indicate that it relies on a general strategy which is not fundamentally different from our own.\footnote{The abseil library does not rely on  Clinger's fast path when parsing numbers. It also uses less accurate product computation.}
Even so, our approach is generally twice as fast as the abseil library and up to five times faster than what the standard library offers.
We find that for the integer test, netlib is superior to all other alternatives (including abseil) except for our own. The gap between our approach and netlib when parsing integers is modest (about 20\%).
Overall, our proposed approach is three to five times faster than the \texttt{strtod} function available in the GNU standard library. And it is often more than
twice as fast as the state-of-the-art abseil library.

\begin{table*}[!tb]\centering
\caption{\label{tab:mfloats} Millions of 64-bit floating-point numbers parsed per second under different processor architectures
}
\setlength{\tabcolsep}{5pt}
\subfloat[Intel Skylake (x64) ]{\footnotesize
\begin{tabular}{cccc}\toprule
                      &   canada    &      uniform    & integer\\ \midrule
netlib                 &   9.6     &          10    &    48 \\
d.-conversion      &   9.4    &         10     &    18\\
strtod                 &     9.0 &       9.4  &    20\\
abseil                 &    18  &       19  &    27\\ 
our parser             & 45    &         45 &     61\\
\bottomrule
\end{tabular}
}
\subfloat[AMD Zen~2 (x64)]{\footnotesize
\begin{tabular}{ccc}\toprule
                         canada        & uniform   &  integer\\ \midrule
10        &     11    &    57 \\
9.0    &      9.9   &   24 \\
9.3       &      9.9   &   18 \\
21 &                21      &   30 \\
51   &      52    &    70 \\
\bottomrule
\end{tabular}
}\\
\subfloat[Ampere Skylark (ARM, aarch64)]{\footnotesize
\begin{tabular}{cccc}\toprule
                        & canada        &  uniform   &  integer\\ \midrule
netlib                  &  8.1        &    8.7    &     23 \\
d.-conversion       &  5.4          &    5.8   &     12 \\ 
strtod                  &   3.9    &   4.2        &   8.7 \\
abseil                  &  9.1    &  9.4        &   13 \\ 
our parser              & 22      &      21     &    26 \\
\bottomrule
\end{tabular}
}
\subfloat[IBM POWER 9]{\footnotesize
\begin{tabular}{ccc}\toprule
                         canada &      uniform    & integer\\ \midrule
9.0     &   10       &    39 \\
5.8       &   6.4      &   18 \\
4.8       &   5.3      &   12 \\
12       &   12       &   17 \\ 
42        &    39     & 46 \\
\bottomrule
\end{tabular}
}

\end{table*}

To understand our good results,  we look and the number of instructions and cycles per number for one representative dataset (uniform) 
and for the AMD Zen~2 processor. See Table~\ref{tab:mfloatsamd}. As expected, we use half as many instructions on average as the
abseil library. We find interesting that we use only about three times fewer instructions than the \texttt{strtod} function, but
5.6 times fewer cycles. Our approach causes almost no branch mispredictions, in contrast with Gay's netlib library.
Similarly, while we retire 4.2~instructions per cycle, Gay's netlib library is limited at 2.2~instructions per cycle.
To summarize, our approach uses fewer instructions, generates fewer branch mispredictions and retires more instructions per cycle.

To identify our bottleneck, we run the parsing routine while skipping the conversion from a decimal significand
and exponent to the standard decimal form. Instead, we sum the decimal significand and the exponent and return
the result as a simulated floating-point value. We find that we save only about a quarter of the number of instructions
and a quarter of the time (cycles). In other words, our decimal-to-binary routine is so efficient that it only
uses about a quarter of our computational time. Most of the time goes into parsing the input string and converting it
to a decimal significand and exponent.

 \begin{table}[!tb]\centering
\caption{\label{tab:mfloatsamd} Instructions, mispredicted branches and cycles per 64-bit floating-point number in the uniform model on the
AMD Zen~2 processor. We also provide the number of instructions per cycle. The ``just string'' row corresponds to our parser but without
the final decimal to binary conversion.}
\setlength{\tabcolsep}{5pt}

\begin{tabular}{ccccc}\toprule
                      &    Instructions & mispredictions & cycles & instructions/cycle\\ \midrule
netlib                &    740                &  4.1       &     330      &     2.2 \\
double-conversion     &    1100               &  1.7       &    380       &     3.0 \\
strtod                &    1100               &  0.7      &     370       &     3.0 \\
abseil                &    600               &   0.5     &       160      &     3.8 \\ 
our parser            &    280              & 0.01       &    66          &     4.2 \\
(just string) &    215              & 0.00       &    46         &     4.7 \\
\bottomrule
\end{tabular}
\end{table}

Our results using 32-bit numbers are similar. To ease comparison, we produce exactly the same numbers strings as in the 64-bit case. We replace the \texttt{strtod} function with the equivalent \texttt{strtof} function.
We present the result in Table~\ref{tab:mfloatsamd32}. It suggests that there is little speed benefit in reading numbers as 32-bit floating-point numbers instead of 64-bit floating-point numbers given the same input strings. The result does not surprise us given that we rely on the same algorithm.

 \begin{table}[!tb]\centering
\caption{\label{tab:mfloatsamd32} Instructions, mispredicted branches and cycles per 32-bit floating-point number in the uniform model on the
AMD Zen~2 processor. We also provide the number of instructions per cycle.}
\setlength{\tabcolsep}{5pt}

\begin{tabular}{ccccc}\toprule
                      &    Instructions & mispredictions & cycles & instructions/cycle\\ \midrule
strtof                &    1100               &  0.7      &     350       &     3.1 \\
abseil                &    600               &   0.5     &       170      &     3.6 \\ 
our parser            &    280              & 0.00       &    64         &     4.3 \\
\bottomrule
\end{tabular}
\end{table}

We find it interesting to represent the parsing speed in terms of bytes per second. 
On the canada dataset using the AMD~Zen2 system, our parser exceeds \SI{1}{\gibi\byte\per\second} (\SI{1080}{\mebi\byte\per\second}). It is $2.5$~times faster than the fastest competitor (abseil) and $5$~times faster than the other parser. See Fig.~\ref{fig:serializedspeed}.
For the synthetic dataset, we use the
concise number serialization instead of relying on a fixed number of digits, to avoid
overestimating the parsing speed. Our parser runs at almost over  \SI{900}{\mebi\byte\per\second} compared to less than \SI{200}{\mebi\byte\per\second} for the \texttt{strtod} function. If we serialize the
numbers so that they use a fixed number of digits (17), we reach higher speeds: our parser exceeds \SI{1}{\gibi\byte\per\second} (not shown).

\begin{figure}\centering
\subfloat[canada]{
\begin{tikzpicture}[scale = 0.7]
\begin{axis}[ 
    ybar,
     bar width=25pt,
    xtick={1,...,5},
    xticklabels from table = {parsingspeedcanada.dat}{label},
    x tick label style = {rotate=0},
    grid=major,
        ymax=1100,
        ymin=0,
    xtick align=inside,
    every axis x label/.style={at={(ticklabel cs:0.5)},anchor=near ticklabel},
    ylabel=throughput (\si{\mebi\byte\per\second}),
    axis lines*=left,
    every axis y label/.style={at={(ticklabel cs:0.5)},rotate=90,anchor=near ticklabel}
]

\addplot[ybar,fill=blue] table [ 
    x=Wert, 
    y=value,
] {parsingspeedcanada.dat} ;

\end{axis} 
\end{tikzpicture} 
}
\subfloat[uniform]{
\begin{tikzpicture}[scale = 0.7]
\begin{axis}[ 
    ybar,
     bar width=25pt,
    xtick={1,...,5},
    xticklabels from table = {parsingspeed.dat}{label},
    x tick label style = {rotate=0},
    grid=major,
        ymax=1100,
        ymin=0,
    xtick align=inside,
    every axis x label/.style={at={(ticklabel cs:0.5)},anchor=near ticklabel},
    ylabel=throughput (\si{\mebi\byte\per\second}),
    axis lines*=left,
    every axis y label/.style={at={(ticklabel cs:0.5)},rotate=90,anchor=near ticklabel}
]

\addplot[ybar,fill=blue] table [ 
    x=Wert, 
    y=value,
] {parsingspeed.dat} ;

\end{axis} 
\end{tikzpicture} 
}
\caption{\label{fig:serializedspeed} Parsing speed for the canada dataset and for random 64-bit floating-point number in the uniform model, serialized concisely, on the
AMD Zen~2 processor.}
\end{figure}
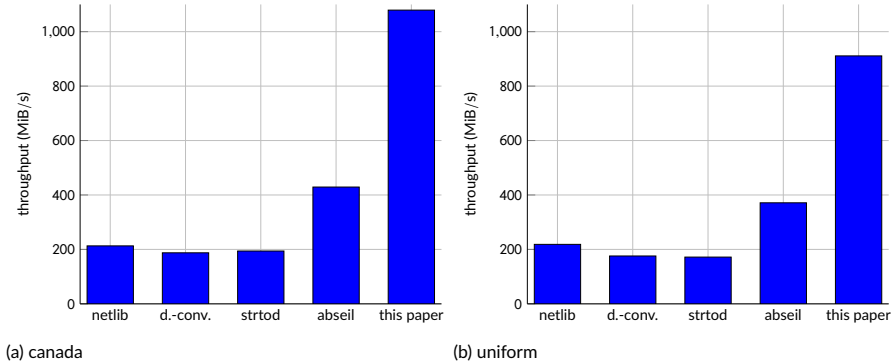

In Table~\ref{tab:percentage}, we provide statistics regarding which code paths are used by different datasets.
The integer dataset is entirely covered by Clinger's fast path. It explains why our performance on this dataset
is similar to the netlib approach, since we rely on essentially the same algorithm. For both the canada
and uniform dataset, most of the processing falls on our parser as opposed to Clinger's fast path.
After initially parsing the input string, 
our fast algorithm begins with one or two multiplications between the decimal significand and looked up table values.
We observe that a single multiplication is all that is necessary in most cases. In our experiments, we never need
to fall back on a higher-precision approach.

 \begin{table}[!tbh]\centering
\caption{\label{tab:percentage} Code path frequencies for different datasets using our parser. The percentages are relative to the number of input number strings.}

\begin{tabular}{p{0.4\textwidth}ccc}\toprule
                      &   canada    &        uniform    & integer\\ \midrule
Clinger's fast path  &   8.8\%     &        0\%    &    100\% \\
our path             &      91.2\% &         100\%       &    0\% \\
two multiplications   &     0.6\%  &       0.66\%     &    0\%\\
\bottomrule
\end{tabular}
\end{table}

\paragraph{Many digits} We designed our algorithm for the scenario where numbers are serialized to strings using no more than 17~digits. However, we can not always ensure that such a reasonable limit is respected. To test the case where we have many more than 17~digits, we create big integer values by serializing three randomly selected 64-bit integers in sequence. On the AMD Zen~2 system, we find that our parser exceeds \SI{1100}{\mebi\byte\per\second}. The abseil library achieves similar speeds \SI{910}{\mebi\byte\per\second} which is more than twice as fast as Gay's netlib \SI{390}{\mebi\byte\per\second}. The \texttt{strtod} function is limited to \SI{110}{\mebi\byte\per\second}.

\paragraph{Visual Studio} Unfortunately, we are not aware of a standard implementation of the \texttt{from\_chars} function under Linux. However,
Microsoft provides one such fast function as part of its Visual Studio~2019 system. We use the latest
available Microsoft C++ compiler (19.26.28806 for x64). We compile in release mode with the flags \texttt{/O2 /Ob2 /DNDEBUG}.
These results under Windows are generally comparable to our Linux results. See Table~\ref{tab:windows}. Microsoft's \texttt{from\_chars} function is faster
than its \texttt{strtod} function. However, our parser is several times faster than Microsoft's \texttt{from\_chars} function.

\begin{table*}[!tb]\centering
\caption{\label{tab:windows} Millions of 64-bit floating-point numbers parsed per second under a  \SI{4.2}{\GHz}  Intel~7700K processor using Visual Studio~2019
}
\begin{tabular}{cccc}\toprule
                      &   canada    &      uniform    & integer\\ \midrule
netlib                 &  20     &          18    &    48 \\
d.-conversion      &       10    &         10     &    18\\
strtod                 &   6.0 &         5.8  &    15\\
from\_chars             &  6.7  &         7.2  &    22\\
abseil                 &  16  &         15  &    22\\   
our parser             &  37    &         48 &     60\\
\bottomrule
\end{tabular}

\end{table*}

\paragraph{Apple M1~Processor} 
In November~2020, Apple released laptops with a novel ARM 3.2\,GHz processor (M1). 
The M1 processor has 8~instruction decoders compared to only 4~decoders on most x64 processors.
Though we would normally avoid benchmarking on a laptop due to potential frequency throttling, we found consistent run-to-run results (within 1\%) and a low margin of error (within 1\%).
We compiled our benchmark software on such a laptop using Apple's LLVM clang compiler (Apple clang version 12.0.0 using the flags \texttt{-O3 -DNDEBUG}). We present our throughput results  in Fig.~\ref{fig:serializedspeedm1}.  Our parser reaches \SI{1.5}{\gibi\byte\per\second} on the uniform dataset. On the Apple platform, the \texttt{strtod} function is several times slower than any other number parser. Other parsers (netlib, double-conversion and abseil) are about three times slower in these tests.

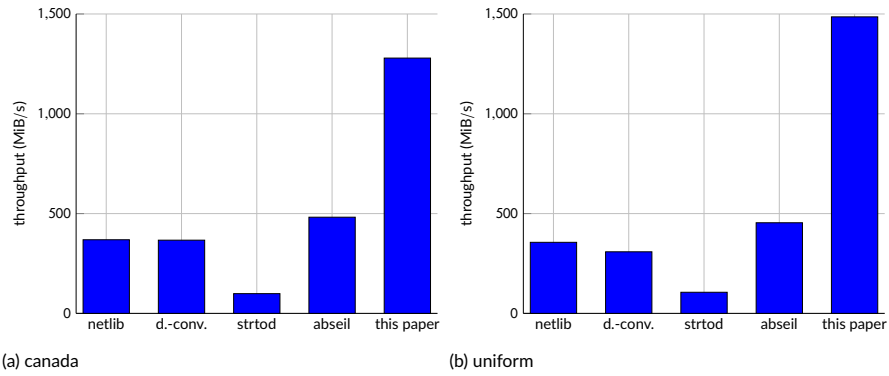
\begin{figure}\centering
\subfloat[canada]{
\begin{tikzpicture}[scale = 0.7]
\begin{axis}[ 
    ybar,
     bar width=25pt,
    xtick={1,...,5},
    xticklabels from table = {m1parsingspeedcanada.dat}{label},
    x tick label style = {rotate=0},
    grid=major,
        ymax=1500,
        ymin=0,
    xtick align=inside,
    every axis x label/.style={at={(ticklabel cs:0.5)},anchor=near ticklabel},
    ylabel=throughput (\si{\mebi\byte\per\second}),
    axis lines*=left,
    every axis y label/.style={at={(ticklabel cs:0.5)},rotate=90,anchor=near ticklabel}
]

\addplot[ybar,fill=blue] table [ 
    x=Wert, 
    y=value,
] {m1parsingspeedcanada.dat} ;

\end{axis} 
\end{tikzpicture} 
}
\subfloat[uniform]{
\begin{tikzpicture}[scale = 0.7]
\begin{axis}[ 
    ybar,
     bar width=25pt,
    xtick={1,...,5},
    xticklabels from table = {m1parsingspeed.dat}{label},
    x tick label style = {rotate=0},
    grid=major,
        ymax=1500,
        ymin=0,
    xtick align=inside,
    every axis x label/.style={at={(ticklabel cs:0.5)},anchor=near ticklabel},
    ylabel=throughput (\si{\mebi\byte\per\second}),
    axis lines*=left,
    every axis y label/.style={at={(ticklabel cs:0.5)},rotate=90,anchor=near ticklabel}
]

\addplot[ybar,fill=blue] table [ 
    x=Wert, 
    y=value,
] {m1parsingspeed.dat} ;

\end{axis} 
\end{tikzpicture} 
}
\caption{\label{fig:serializedspeedm1} Parsing speed for the canada dataset and for random 64-bit floating-point number in the uniform model, serialized concisely, on the
Apple~M1 processor.}
\end{figure}

\section{Conclusion}

Parsing floating-point numbers from strings is a fundamental operation supported by the
standard library of almost all programming languages. Our results suggest that widely used implementations
might be several times slower than needed on modern 64-bit processors. When the input strings are retrieved from
disks or networks with gigabytes per second in bandwidth, a faster approach should be beneficial.

We expect that more gains are possible mostly in how we parse the input strings into a decimal significand and exponent. 
For example, we could use advanced processor instructions such as SIMD instructions~\cite{langdale2019parsing}. 


It also be possible to accelerate the processing by relaxing correctness conditions: e.g., the parsing could be only exact up to an error in the last digit. However, we should be mindful of the potential problems that arise when different software components parse the same numbers to different binary values.

Floating-point numbers may be stored in binary form and accessed directly without parsing. However, some engineers
prefer to rely on text formats. Hexadecimal floating-point numbers (Appendix~\ref{sec:hexfloat}) may provide a convenient alternative for greater speed in such cases.



\section*{Acknowledgements}

Our work benefited especially from exchanges with 
M.~Eisel who motivated the original research with his key insights. We thank N.~Tao who provided invaluable feedback and who contributed an earlier and simpler
version of this algorithm to the Go standard library.  
Our fallback implementation includes code adapted from Google Wuffs, a memory-safe programming language, which was published  under the Apache 2.0 license. 
To our knowledge, the fast path for long numbers was first implemented  by R.~Oudompheng for the Go standard library.
We thank A.~Milovidov for his feedback regarding benchmarking.
We are grateful to W.~Mu\l{}a for his thorough review of an early manuscript: his comments helped us improve the document significantly.
We thank I.~Smirnov for his feedback on benchmarking statistics.
We thank P.~Cawley for his feedback on the manuscript.
\bibliography{fastfloat}

\appendix
\section{Multiplicative Inverses}

Given an odd 64-bit integer $x$, there is a
unique integer $y$ such that $x \times y \bmod 2^{64} = 1$. We refer to $y$ as the \emph{multiplicative inverse} of $x$~\cite{dumas2013newton}.
Fig.~\ref{fig:multiinverse} presents an efficient C++ function to compute the multiplicative inverse of 64-bit odd integers. It relies on five successive calls to a function involving two integer multiplications.
\begin{figure}
\lstset{escapechar=@,style=customc}
\begin{lstlisting}
uint64_t f64(uint64_t x, uint64_t y) {  return  y * ( 2 - y * x ); }
uint64_t findInverse64(uint64_t x) {
   uint64_t y = x; y = f64(x,y); y = f64(x,y); y = f64(x,y);
   y = f64(x,y); y = f64(x,y); return y;
}
\end{lstlisting}
\caption{C++ function (\texttt{findInverse64}) to compute the multiplicative inverse of an odd 64-bit integer using Newton's method\label{fig:multiinverse}}
\end{figure}

\section{Table Generation Script}
\label{appendix:table}
Fig.~\ref{fig:python} provides a convenient Python script to general all relevant reciprocal and normalized powers of five. In practice, each 128-bit value may be stored as two 64-bit words.

\begin{figure}[hbt]
\lstset{escapechar=@,style=custompython}
\begin{lstlisting}
for q in range(-342,-27):
    power5 = 5**-q
    z = 0
    while( (1<<z) < power5) : z += 1
    b = 2 * z + 2 * 64
    c = 2 ** b // power5 + 1
    while(c >= (1<<128)): c //= 2
    print(c)    

for q in range(-27,0):
    power5 = 5**-q
    z = 0
    while( (1<<z) < power5) : z += 1
    b = z + 127
    c = 2 ** b // power5 + 1
    print(c)
    
for q in range(0,308+1):
    power5 = 5**q
    while(power5 < (1<<127)) : power5 *= 2
    while(power5 >= (1<<128)): power5 //= 2
    print(power5)
\end{lstlisting}
\caption{Python script to print out all 128-bit reciprocals ($q\in [-342,0)$) and all 128-bit truncated powers of five ($q \in [0,308)$).\label{fig:python}}
\end{figure}

\section{Hexadecimal Floating-Point Numbers}
\label{sec:hexfloat}

It could be  convenient to represent floating-point numbers using the  hexadecimal floating-point notation. The hexadecimal notation may provide an exact ASCII string representation of the binary floating-point number. It makes it relatively easy to provide an unambiguous string that should always be parsed to the same binary value. Furthermore, the parsing and serialization speeds could be much higher. The main downsides are that human beings may find such strings harder to understand and that they are not natively supported in all mainstream programming languages.

The  hexadecimal floating-point notation is supported in the C (C99), C++ (C++17), Swift, Java, Julia and Go programming languages.
As in the usual hexadecimal notation for integers,
we start the string with \texttt{0x} followed by the significand in hexadecimal form. Each hexadecimal character (0--9, A--F) represents 4~bits (a \emph{nibble}).  Instead of writing the exponential part in full (e.g., $\times 2^{4}$ or  $\times 2^{-4}$), we append the suffix \texttt{p} followed by the exponent (e.g., \texttt{p4} or \texttt{p-4}). 
Optionally, we can add an hexadecimal point in the significand. With a decimal point, we interpret the decimal fraction by dividing it by the appropriate power of ten. E.g, we write $1.45= 145 / 10^2$. The hexadecimal point works similarly. Thus \texttt{0x1.FCp17} means $\mathtt{0x1FC} / 16^2 \times 2^{17}$ or \num{260096} where we divide $\mathtt{0x1FC}$ by  $16^2$ because there are
two nibbles after the binary point. 
When the value is a normal 64-bit floating-point number, the significand can be expressed as a most significant 1 followed by up to 52~bits, or 13~hexadecimal character. Thus 
\num{9000000000000000} can be written as \texttt{0x1.ff973cafa8p+52}. The mass of the Earth in kilogram ($5.972\times 10^{24}$) is \texttt{0x1.3c27b13272fb6p+82}.

\section{String-Parsing Functions in C++}
\label{appendix:parsing-string}

Fig.~\ref{fig:parsestring} illustrates the computation of the decimal significand with pseudo-C++ code. We omit the code necessary to check whether there are leading spaces and sign characters (\texttt{+} or  \texttt{-}) and other error checks. We must further parse an eventual exponent preceded by the characters \texttt{e} or  \texttt{E}. Moreover, we must also check whether we had more than 19~digits in the decimal significand. Thus our actual code is slightly more complex.
Fig.~\ref{fig:pseudocppswar} presents SWAR~\cite{fisher1998compiling} functions to  check all at once whether a sequence of 8~digits is available and to compute the corresponding decimal integer.

\begin{figure}[!tbh]
\lstset{escapechar=@,style=customc}
\begin{lstlisting}
  const char *p = // points at the beginning of the string
  const char *pend = // points at the end of the string
  int64_t exponent = 0; // exponent
  uint64_t i = 0; // significand 
  while ((p != pend) && is_integer(*p)) {  i = 10 * i + uint64_t(*p - '0');  ++p; }
  if ((p != pend) && (*p == '.')) {
    ++p; const char *first_after_period = p;
    if ((p + 8 <= pend) && is_made_of_eight_digits(p)) {
      i = i * 100000000 + parse_eight_digits(p); p += 8;
      if ((p + 8 <= pend) && is_made_of_eight_digits(p))
      {  i = i * 100000000 + parse_eight_digits(p); p += 8; }
    }
    while ((p != pend) && is_integer(*p)) { uint8_t digit = uint8_t(*p - '0'); ++p; 
    i = i * 10 + digit; }
    exponent = first_after_period - p;
  }
\end{lstlisting}
\caption{Simplified pseudo-C++ code to compute the decimal significand from an ASCII string\label{fig:parsestring}}
\end{figure}

\begin{figure}[!tbh]
    \centering
\lstset{escapechar=@,style=customc}
\begin{lstlisting}
bool is_made_of_eight_digits(const char *chars) {
  uint64_t val; memcpy(&val, chars, 8);
  return !((((val + 0x4646464646464646)  | (val - 0x3030303030303030)) 
       & 0x8080808080808080)); 
}

uint32_t parse_eight_digits(const char *chars)  {
  uint64_t val; memcpy(&val, chars, sizeof(uint64_t));
  val = (val & 0x0F0F0F0F0F0F0F0F) * 2561 >> 8;
  val = (val & 0x00FF00FF00FF00FF) * 6553601 >> 16;
  return uint32_t((val & 0x0000FFFF0000FFFF) * 42949672960001 >> 32);
}

\end{lstlisting}
\caption{C++ functions to check whether 8~ASCII characters are made of digits, and to convert them to an integer value under a little-endian system\label{fig:pseudocppswar}}
\end{figure}

\section{Benchmarks in Rust}
\label{appendix:rust_results}

Our C++ implementation and benchmarks have been ported to Rust by 
I.~Smirnov.\footnote{\url{https://github.com/aldanor/fast-float-rust}} Unlike our C++ implementation, it does not attempt to skip leading white spaces, but there are otherwise few differences.  This Rust port allows us to compare against 
a popular Rust number processing library  (lexical\footnote{\url{https://docs.rs/lexical/5.2.0/lexical/}---v5.2.0}) as well as the standard Rust library (\texttt{from\_str}). Using Rust~1.49 on our AMD~Rome (Zen~2) system, we get the following results on the canada dataset: the standard Rust library is limited to \SI{92}{\mebi\byte\per\second}, lexical library achieves 
\SI{280}{\mebi\byte\per\second} while the Rust
port of our library achieves \SI{670}{\mebi\byte\per\second}.
On the Apple~M1 system, we get \SI{130}{\mebi\byte\per\second} (standard library),
\SI{370}{\mebi\byte\per\second} (lexical) and 
\SI{1200}{\mebi\byte\per\second} (Rust port).
The tests are repeated 1000~times and the difference between the best speed and the median speed is low on our test systems (less than 1\%).

\end{document}